\documentclass[12pt]{article}

\usepackage[margin=1.0in]{geometry}
 
\usepackage{graphicx, subcaption}      
\usepackage{amssymb, amsmath, algorithm2e, setspace}
\usepackage{tikz}
\usetikzlibrary{patterns,angles,quotes,arrows}
\usepackage[shortlabels]{enumitem}
\usepackage{amsthm}
\usepackage{url}
\newtheoremstyle{mystyle}
  {}
  {}
  {}
  {}
  {\bfseries}
  {:}
  { }
  {}

\theoremstyle{mystyle}
\newtheorem{theorem}{Theorem}
\newtheorem{problem}{Problem}
\newtheorem{definition}{Definition}
\newtheorem{proposition}{Proposition}
\newtheorem{lemma}{Lemma}

\newtheorem{remark}{Remark}
\newtheorem{corollary}{Corollary}
\usepackage[numbers]{natbib}

\providecommand{\keywords}[1]{\textbf{\textit{Index terms---}} #1}
\usepackage{cite}
\usepackage{appendix}

\DeclareMathOperator{\rank}{rank}
\DeclareMathOperator{\interior}{int}

\begin{document}

\title{Quantitative Resilience of Generalized Integrators\thanks{This work was supported by an Early Stage Innovations grant from NASA’s Space Technology Research Grants Program, grant no. 80NSSC19K0209. This material is partially based upon work supported by the United States Air Force AFRL/SBRK under contract no. FA864921P0123.}}

\author{Jean-Baptiste Bouvier\thanks{Department of Aerospace Engineering and Coordinated Science Laboratory, University of Illinois at Urbana-Champaign, Urbana, IL 61801, USA.
  (bouvier3@illinois.edu, mornik@illinois.edu).}
\and Kathleen Xu\footnotemark[2] \and Melkior Ornik\footnotemark[2]
}

\maketitle

\begin{abstract}
    To design critical systems engineers must be able to prove that their system can continue with its mission even after losing control authority over some of its actuators. Such a malfunction results in actuators producing possibly undesirable inputs over which the controller has real-time readings but no control. By definition, a system is resilient if it can still reach a target after a partial loss of control authority. However, after such a malfunction, a resilient system might be significantly slower to reach a target compared to its initial capabilities. To quantify this loss of performance we introduce the notion of quantitative resilience as the maximal ratio of the minimal times required to reach any target for the initial and malfunctioning systems. Naive computation of quantitative resilience directly from the definition is a complex task as it requires solving four nested, possibly nonlinear, optimization problems. The main technical contribution of this work is to provide an efficient method to compute quantitative resilience of control systems with multiple integrators and nonsymmetric input sets.
    Relying on control theory and on two novel geometric results we reduce the computation of quantitative resilience to a linear optimization problem. We illustrate our method on two numerical examples: a trajectory controller for low-thrust spacecrafts and a UAV with eight propellers.
\end{abstract}

\keywords{reachability, quantitative resilience, linear systems, optimization}

\vspace{5mm}

\textbf{Notice of previous publication:} This manuscript is a substantially extended version of \citep{SIAM_CT} where we remove the assumption of symmetry on the input sets, leading to more general and more complex results, e.g., Theorems~\ref{thm:direction maximizing t} and \ref{thm:computation of r_q}.
This paper also provides several proofs omitted from \citep{SIAM_CT} and tackle systems with multiple integrators. Entirely novel material includes Sections~\ref{section:integrators}, \ref{section:examples}, Appendices~\ref{apx:continuity}, \ref{apx:bar B} and parts of all other sections.

\section{Introduction}
When failure is not an option, critical systems are built with sufficient actuator redundancy \citep{NASA_redundancy} and with fault-tolerant controllers \citep{partial_LOC}.
These systems rely on different methods like adaptive control \citep{actuator_lock, fault_tolerance} or active disturbance rejection \citep{hypersonic_fault_tolerant} in order to compensate for actuator failures.
The study of this type of malfunction typically considers either actuators locking in place \citep{actuator_lock}, actuators losing effectiveness but remaining controllable \citep{partial_LOC, fault_tolerance}, or a combination of both \citep{hypersonic_fault_tolerant}.
However, when actuators are subject to damage or hostile takeover, the malfunction may result in the actuators producing undesirable inputs over which the controller has real-time readings but no control.
This type of malfunction has been discussed in \citep{IFAC} under the name of \emph{loss of control authority} over actuators and encompasses scenarios where actuators and sensors are under attack \citep{actuator_attack}.

In the setting of loss of control authority, undesirable inputs are observable and can have a magnitude similar to the controlled inputs, while in classical robust control the undesirable inputs are not observable and have a small magnitude compared to the actuators' inputs \citep{Tubes1971, Robust}.
The results of \citep{journal_paper} showed that a controller having access to the undesirable inputs is considerably more effective than a robust controller.  

After a partial loss of control authority over actuators, a target is said to be \emph{resiliently reachable} if for any undesirable inputs produced by the malfunctioning actuators there exists a control driving the state to the target \citep{IFAC}. However, after the loss of control the malfunctioning system might need considerably more time to reach its target compared to the initial system. 
Previous work \citep{SIAM_CT} introduced the concept of quantitative resilience for control systems in order to measure the delays caused by the loss of control authority over actuators. 
While concepts of quantitative resilience have been previously developed for water infrastructure systems \citep{water_qr} or for nuclear power plants \citep{nuclear_pp_qr}, such concepts only work for their specific application.

In this work we formulate quantitative resilience as the maximal ratio of the minimal times required to reach any target for the initial and malfunctioning systems.
This formulation leads to a nonlinear minimax optimization problem with an infinite number of equality constraints.
Because of the complexity of this problem, a straightforward attempt at a solution is not feasible.
While for linear minimax problems with a finite number of constraints the optimum is reached on the boundary of the constraint set \citep{max-min_programming}, such a general result does not hold in the setting of semi-infinite programming \citep{semi-infinite_programming} where our problem belongs.
However, the fruitful application of the theorems of \citep{liberzon, Neustadt} stating the existence of time-optimal controls combined with the optimization results derived in \citep{Maxmax_Minimax_Quotient_thm} reduces the quantitative resilience of systems with multiple integrators to a linear optimization problem.

As a first step toward the study of quantitative resilience for linear systems we restricted our previous work \citep{SIAM_CT} to driftless linear systems with symmetric input sets. Building on these earlier results we now extend the theory to linear systems with multiple integrators and general input sets.
With these extensions we are able to tackle feedback linearized systems \citep{feedback_linearization}.

The contributions of this paper are threefold.
First, we extend the notion of quantitative resilience from \citep{SIAM_CT} to systems with multiple integrators and nonsymmetric inputs. 
Secondly, we provide an efficient method to compute the quantitative resilience of linear systems with integrators by simplifying a nonlinear problem of four nested optimizations into a linear optimization problem.
Finally, based on quantitative resilience and controllability we establish necessary and sufficient conditions to verify if a system is resilient.

The remainder of the paper is organized as follows. Section~\ref{section:preliminaries} introduces preliminary results concerning resilient systems and quantitative resilience. 
To evaluate this metric we need the minimal time for the system to reach a target before and after the loss of control authority.
We calculate this minimal time for the initial system in Section~\ref{section:initial dynamics} and for the malfunctioning system in Section~\ref{section:malfunctioning dynamics}.
Section~\ref{section:r_q} is the pinnacle of this work as we design an efficient method to compute quantitative resilience.
This metric also allows to assess whether a system is resilient or not, as detailed in Section~\ref{section:resilience conditions}. 
The quantitative resilience of systems with multiple integrators is studied in Section~\ref{section:integrators}.
In Section~\ref{section:examples} our theory is illustrated on a linear trajectory controller for a low-thrust spacecraft and on an unmanned octocopter.
The continuity of a minimum function is proved in Appendix~\ref{apx:continuity}.
We compute the dynamics of the low-thrust spacecraft in Appendix~\ref{apx:bar B}.

\vspace{2mm}

\emph{Notation:} 
For a set $X$ we denote its boundary $\partial X$, its interior $X^\circ := X \backslash \partial X$ and its convex hull $co(X)$.
The set of integers between $a$ and $b$ is $[\![a,b]\!]$, while $\mathbb{N}$ denotes the set of all positive integers. 
The factorial of $k \in \mathbb{N}$ is denoted $k!$.
Let $\mathbb{R}^+ := [0, \infty)$ and we use the subscript $_*$ to exclude zero, for instance $\mathbb{R}_*^+ := (0, \infty)$.
We denote the Euclidean norm with $\| \cdot \|$ and the unit sphere with $\mathbb{S} := \{ x \in \mathbb{R}^n : \|x\| = 1\}$.
The infinity-norm of a vector $x \in \mathbb{R}^n$ is $\| x \|_\infty := \max \big\{ |x_i| : i \in [\![1,n]\!] \big\}$.
The image of a matrix $A \in \mathbb{R}^{n \times m}$ is denoted $Im(A) \subset \mathbb{R}^n$, its rank is $rank(A) = \dim Im(A) \leq n$.
For integrable piecewise continuous functions $f: \mathbb{R} \rightarrow \mathbb{R}^n$, the $\mathcal{L}_2$-norm is defined as $\| f \|_{\mathcal{L}_2}^2 := \int_{t\, \in\, \mathbb{R}} \|f(t)\|^2\, dt$, and the $\mathcal{L}_\infty$-norm is defined as $\| f \|_{\mathcal{L}_\infty} := \underset{t\, \in\, \mathbb{R}}{\sup}\ \|f(t)\|_\infty$.
For $k \in \mathbb{N}$, the $k^{th}$ derivative of a function $f$ is denoted as $f^{(k)}$.

\section{Preliminaries and Problem Statement}\label{section:preliminaries}

We are interested by generalized $k^{th}$ order integrators in $\mathbb{R}^n$, i.e.,
\begin{equation}\label{eq:system order k}
    x^{(k)}(t) = \bar{B}\bar{u}(t), \quad \text{with} \quad \bar{u} \in \bar{U}, \quad x(0) = x_0, \quad \text{and} \quad x^{(l)}(0) = 0,
\end{equation}
for all $l \in [\![1,k-1]\!]$ and $k \in \mathbb{N}$. Matrix $\bar{B} \in \mathbb{R}^{n \times (m+p)}$ is constant. Let $\bar{u}^{min} \in \mathbb{R}^{m+p}$ and $\bar{u}^{max} \in \mathbb{R}^{m+p}$ be the bounds on the inputs so that the set of allowable controls is 
\begin{equation}\label{eq:U bar}
    \bar{U} := \big\{  \bar{u} \in \mathcal{L}_\infty(\mathbb{R}^+, \mathbb{R}^{m+p}) :  \bar{u}_i(t) \in [\bar{u}_i^{min}, \bar{u}_i^{max}]\ \text{for}\ t \geq 0\ \text{and}\ i \in [\![1,m+p]\!] \big\}.
\end{equation}
After a malfunction, the system loses control authority over $p$ of its $m+p$ initial actuators. Because of the malfunction, the initial control input $\bar{u}$ is split into the remaining controlled inputs $u$ and the undesirable inputs $w$.
Similarly, we split the control bounds $\bar{u}^{min}$ and $\bar{u}^{max}$ into $u^{min} \in \mathbb{R}^m$, $u^{max} \in \mathbb{R}^m$ and $w^{min} \in \mathbb{R}^p$, $w^{max} \in \mathbb{R}^p$.
Without loss of generality we always consider the columns $C$ representing the malfunctioning actuators to be at the end of $\bar{B}$. We split the control matrix accordingly: $\bar{B} = \big[ B\ C\big]$. Then, the initial conditions are the same as in \eqref{eq:system order k} but the dynamics become
\begin{align}
    & x^{(k)}(t) = Bu(t) + Cw(t), \qquad u \in U, \quad w \in W, \quad \text{with} \label{eq:splitted order k} \\
    U &:= \big\{ u \in \mathcal{L}_\infty(\mathbb{R}^+, \mathbb{R}^m) :  u_i(t) \in [u_i^{min}, u_i^{max}]\ \text{for}\ t \geq 0\ \text{and}\ i \in [\![1,m]\!] \big\} \quad \text{and} \label{eq:function set U}\\
    W &:= \big\{ w \in \mathcal{L}_\infty(\mathbb{R}^+, \mathbb{R}^p) :  w_i(t) \in [w_i^{min}, w_i^{max}]\ \text{for}\ t \geq 0\ \text{and}\ i \in [\![1,p]\!] \big\}.\label{eq:function set W}
\end{align}

We will use the concept of \emph{controllability} of \citep{liberzon}.
\begin{definition}
    A system following dynamics \eqref{eq:system order k} is \emph{controllable} if for all target $x_{goal} \in \mathbb{R}^n$ there exists a control $\bar{u} \in \bar{U}$ and a time $T$ such that $x(T) = x_{goal}$.
\end{definition}
We recall here the definition of the \emph{resilience} of a system introduced in \citep{journal_paper}.
\begin{definition}
    A system following dynamics \eqref{eq:system order k} is \emph{resilient} to the loss of $p$ of its actuators corresponding to the matrix $C$ as above if for all undesirable inputs $w \in W$ and all target $x_{goal} \in \mathbb{R}^n$ there exists a control $u \in U$ and a time $T$ such that the state of the system \eqref{eq:splitted order k} reaches the target at time $T$, i.e., $x(T) = x_{goal}$.
\end{definition}

Previous efforts \citep{IFAC, journal_paper} assumed that the $\mathcal{L}_2$-norm of the inputs is constrained. In this work, we instead consider $\mathcal{L}_\infty$ bounds because of their broad use in applications. Therefore, most of the resiliency conditions of \citep{IFAC, journal_paper} do not directly apply here. 
We will establish a simple necessary condition for this new setting.

\begin{proposition}\label{prop:resilient full rank}
    If the system \eqref{eq:system order k} is resilient to the loss of $p$ actuators, then the system $x^{(k)}(t) = B u(t)$ is controllable.
\end{proposition}
\begin{proof}
    Let $y \in \mathbb{R}^n$, $x_{goal} := y + x_0 \in \mathbb{R}^n$ and $w \in W$ such that $w(t) = 0$ for all $t \geq 0$. Since the system is resilient, there exist $u \in U$ and $t_1 \geq 0$ such that $x_{goal} = x(t_1)$. Since $x^{(l)}(0) = 0$ for $l \in [\![1,k-1]\!]$ and $B$ is constant we can write
    \begin{equation*}
        x(t_1) = \hspace{-1mm} \int_0^{t_1} \hspace{-4mm} \hdots \hspace{-1mm} \int_0^{t_k} \hspace{-3mm} x^{(k)}(t_{k+1})\, dt_{k+1} \hdots dt_2 + x_0 = Bz + x_0, \hspace{2mm} z := \hspace{-1mm} \int_0^{t_1} \hspace{-4mm} \hdots \hspace{-1mm} \int_0^{t_k} \hspace{-3mm} u(t_{k+1})\, dt_{k+1}\hdots dt_2,
    \end{equation*}
    with $z \in \mathbb{R}^m$ constant.
    Then, $x_{goal} - x_0 = Bz = y \in Im(B)$, so $rank(B) = n$ and $x^{(k)}(t) = B u(t)$ is controllable.  
\end{proof}

While a resilient system is capable of reaching any target after losing control authority over $p$ of its actuators, the time for the malfunctioning system to reach a target might be considerably larger than the time needed for the initial system to reach the same target.
We introduce these two times for the target $x_{goal} \in \mathbb{R}^n$ and the target distance $d := x_{goal} - x_0 \in \mathbb{R}^n$.

\begin{definition}\label{def: T_N^*}
    The \emph{nominal reach time of order $k$}, denoted by $T_{k,N}^*$, is the shortest time required for the state $x$ of \eqref{eq:system order k} to reach the target $x_{goal}$ under admissible control $\bar{u} \in \bar{U}$: 
    \begin{equation}\label{eq:T_k,N^*}
        T_{k,N}^*(d) := \underset{\bar{u}\, \in\, \bar{U}}{\inf} \big\{ T \geq 0 : x(T) - x_0 = d \big\}. 
    \end{equation}
\end{definition}
\begin{definition}\label{def: T_M^*}
    The \emph{malfunctioning reach time of order $k$}, denoted by $T_{k,M}^*$, is the shortest time required for the state $x$ of \eqref{eq:splitted order k} to reach the target $x_{goal}$ under admissible control $u \in U$ when the undesirable input $w \in W$ is chosen to make that time the longest:
   \begin{equation}\label{eq:T_k,M^*}
        T_{k,M}^*(d) := \underset{w\, \in\, W}{\sup} \big\{ \underset{u\, \in\, U}{\inf} \big\{ T \geq 0 : x(T) - x_0 = d \big\} \big\}.
    \end{equation}
\end{definition}

By definition, if the system is controllable, then $T_{k,N}^*(d)$ is finite for all $d \in \mathbb{R}^n$, and if it is resilient, then $T_{k,M}^*(d)$ is finite. 
In light of \eqref{eq:T_k,N^*} and \eqref{eq:T_k,M^*}, $T_{k,N}^*$ and $T_{k,M}^*$ would also depend on the initial conditions $x^{(l)}(0)$ if they were non zero for $l \in [\![1, k-1]\!]$.
It would cause the unnecessary complications discussed in Remark~\ref{rmk: non-zero_initial_conditions}.

\begin{definition}
    The \emph{ratio of reach times of order $k$} in the direction $d \in \mathbb{R}^n$ is
    \begin{equation}\label{eq:t_k(d)}
        t_k(d) :=  \frac{T_{k,M}^*(d)}{T_{k,N}^*(d)}.
    \end{equation}
\end{definition}

After the loss of control, the malfunctioning system \eqref{eq:splitted order k} can take up to $t_k(d)$ times longer than the initial system \eqref{eq:system order k} to reach the target $d + x_0$. Since the performance is degraded by the undesirable inputs, $t_k(d) \geq 1$. We take the convention that $t_k(d) = +\infty$ whenever $T_{k,M}^*(d) = +\infty$, regardless of the value of $T_{k,N}^*(d)$.

\begin{remark}\label{rmk:d=0}
    The case $T_{k,N}^*(d) = T_{k,M}^*(d) = 0$ can only happen when $d = 0$, because $x(0) = x_0 = x_{goal}$. To make this case coherent with \eqref{eq:t_k(d)} and subsequent definitions we choose $\frac{T_{k,N}^*(0)}{T_{k,M}^*(0)} = 1$.
\end{remark}

In order to quantify the resilience of a system, we introduce the following metric.

\begin{definition}
    The \emph{quantitative resilience of order $k$} of system \eqref{eq:splitted order k} is
    \begin{equation}\label{eq:r_k,q}
        r_{k,q} := \frac{1}{\underset{d\, \in\, \mathbb{R}^n}{\sup}\, t_k(d)} = \underset{d\, \in\, \mathbb{R}^n}{\inf} \frac{T_{k,N}^*(d)}{T_{k,M}^*(d)}.
    \end{equation}
\end{definition}

For a resilient system, $r_{k,q} \in (0,1]$. The closer $r_{k,q}$ is to $1$, the smaller is the loss of performance caused by the malfunction.
Quantitative resilience depends on matrices $B$ and $C$, i.e., on the actuators that are producing undesirable inputs.

Computing $r_{k,q}$ naively requires solving four nested optimization problems over continuous constraint sets, with three of them being infinite-dimensional function spaces. A brute force approach to this problem is doomed to fail.

\begin{problem}
    Establish an efficient method to compute $r_{k,q}$.
\end{problem}

We will explore thoroughly the simple case $k = 1$ in the following sections and generalize their results to the case $k \in \mathbb{N}$ in Section~\ref{section:integrators}.
For $k = 1$, the systems \eqref{eq:system order k} and \eqref{eq:splitted order k} simplify into the following linear driftless systems
\begin{align}
    \dot{x}(t) &= \bar{B} \bar{u}(t), \qquad \text{with} \quad \bar{u} \in \bar{U} \qquad \text{and} \quad x(0) = x_0 \in \mathbb{R}^n, \label{eq:system order 1} \\
    \dot{x}(t) &= Bu(t) + Cw(t), \quad \text{with} \quad u \in U, \quad w \in W \quad \text{and} \quad x(0) = x_0 \in \mathbb{R}^n. \label{eq:splitted order 1}
\end{align}
For $k = 1$ we are also able to write the \emph{nominal reach time} $T_N^*$ as 
\begin{equation}\label{eq:nominal reach time}
    T_N^*(d) := \underset{\bar{u}\, \in\, \bar{U} }{\inf} \Big\{ T \geq 0 : \int_0^T \hspace{-3mm} \bar{B} \bar{u}(t)\, dt = d \Big\},
\end{equation}
and the \emph{malfunctioning reach time} $T_M^*$ as
\begin{equation}\label{eq:malfunctioning reach time}
    T_M^*(d) := \underset{w\, \in\, W}{\sup} \Bigg\{ \underset{u\, \in\, U}{\inf} \Big\{ T \geq 0 : \int_0^T \hspace{-3mm} Bu(t) + Cw(t)\, dt = d \Big\} \Bigg\}.
\end{equation}
The \emph{ratio of reach times} in the direction $d \in \mathbb{R}^n$ becomes $t(d) :=  T_M^*(d) / T_N^*(d)$.
The \emph{quantitative resilience} $r_q$ of a system following \eqref{eq:splitted order 1} is then
 \begin{equation}\label{eq:r_q}
     r_q := \frac{1}{\underset{d\, \in\, \mathbb{R}^n}{\sup}\, t(d)} = \underset{d\, \in\, \mathbb{R}^n}{\inf}\ \frac{T_N^*(d)}{T_M^*(d)}.
 \end{equation}
Then, $T_{1,N}^*(d) = T_N^*(d)$, $T_{1,M}^*(d) = T_M^*(d)$, $t_1(d) = t(d)$ for all $d \in \mathbb{R}^n$ and $r_{1,q} = r_q$.

We now discuss the information setting in the malfunctioning system.
The resilience framework of \citep{IFAC, journal_paper} assumes that $u$ has only access to the past and current values of $w$, but not to their future. Then, the optimal control $u^*$ in \eqref{eq:malfunctioning reach time} cannot anticipate a truly random undesirable input $w$. Hence, this strategy is not likely to result in the global time-optimal trajectory of Definition~\ref{def: T_M^*}.

In fact, there would be no single obvious choice for $u^*\big(t, w(t)\big)$, rendering $T_M^*$ ill-defined and certainly not time-optimal, whereas $T_N^*$ is time-optimal. In this case, our concept of quantitative resilience becomes meaningless.
The work \citep{Borgest} states that to calculate $u^*$ without future knowledge of $w^*$ the only technique is to solve the intractable Isaac's equation. Thus, the paper \citep{Borgest} derives only suboptimal solutions and concludes that its practical contribution is minimal.

Instead, we follow \citep{Sakawa} where the inputs $u^*$ and $w^*$ are both chosen to make the transfer from $x_0$ to $x_{goal}$ time-optimal in the sense of Definition~\ref{def: T_M^*}.
The controller knows that $w^*$ will be chosen to make $T_M^*$ the longest. Thus, $u^*$ is chosen to react optimally to this worst undesirable input. Then, $w^*$ is chosen, and to make $T_M^*$ the longest, it is the same as the controller had predicted. Hence, from an outside perspective it looks as if the controller knew $w^*$ in advance, as reflected by \eqref{eq:T_k,M^*}.

We will prove in the following sections that with this information setting $w^*$ is constant. Then, the controller can more easily and more reasonably predict what is the worst $w^*$ and build the adequate $u^*$. With these two input signals, $T_M^*$ is time-optimal in the sense of Definition~\ref{def: T_M^*} and can be meaningfully compared with $T_N^*$ to define the quantitative resilience of control systems.

\section{Dynamics of the Initial System}\label{section:initial dynamics}

We start with the initial system of dynamics \eqref{eq:system order 1} and aim to calculate the nominal reach time $T_N^*$. We introduce the constant input set $\bar{U}_c := \big\{ \bar{u} \in \mathbb{R}^{m+p} : \bar{u}_i \in [\bar{u}_i^{min}, \bar{u}_i^{max}]\ \text{for}\ i \in [\![1,m+p]\!] \big\}$.

\begin{proposition}\label{prop:unperturbed time}
    For a controllable system \eqref{eq:system order 1} and $d = x_{goal} - x_0 \in \mathbb{R}^n$, the infimum $T_N^*(d)$ of \eqref{eq:nominal reach time} is achieved with a constant control input $\bar{u}^* \in \bar{U}_c$.
\end{proposition}

\begin{proof}
    Dynamics \eqref{eq:system order 1} are linear in $x$ and $\bar{u}$. Set $\bar{U}$ defined in \eqref{eq:U bar} is convex and compact. 
    The system is controllable, so $x_{goal}$ is reachable.
    The assumptions of Theorem 4.3 of \citep{liberzon} are satisfied, leading to the existence of a time optimal control $\hat{u} \in \bar{U}$.
    Thus, the infimum in \eqref{eq:nominal reach time} is a minimum and $\int_0^{T_N^*} \bar{B} \hat{u}(t)\, dt = d$. If $d = 0$, then according to Remark~\ref{rmk:d=0}, $T_N^* = 0$ and we take $\bar{u}^* = 0$ so that $\bar{B}\bar{u}^* T_N^* = d$.
    Otherwise, $T_N^* > 0$, so we can define the constant vector $\bar{u}^* := \frac{1}{T_N^*} \int_0^{T_N^*} \hat{u}(t)\, dt \in \mathbb{R}^{m+p}$.
    Since $\hat{u} \in \bar{U}$, for $t \geq 0$ and $i \in [\![1,m+p]\!]$ we have $\bar{u}_i^{min} \leq \hat{u}_i(t) \leq \bar{u}_i^{max}$. Then,
    \begin{equation*}
        \frac{1}{T_N^*} \int_0^{T_N^*} \hspace{-2mm} \bar{u}_i^{min}\, dt = \bar{u}_i^{min} \leq \frac{1}{T_N^*} \int_0^{T_N^*} \hspace{-2mm} \hat{u}_i(t)\, dt = \bar{u}_i^* \leq \frac{1}{T_N^*} \int_0^{T_N^*} \hspace{-2mm} \bar{u}_i^{max}\, dt = \bar{u}_i^{max}.
    \end{equation*}
    Thus, $\bar{u}^* \in \bar{U}_c$. Additionally, $\int_0^{T_N^*} \bar{B} \bar{u}^*\, dt = \bar{B} \bar{u}^* T_N^* = d$. 
\end{proof}

Following Proposition~\ref{prop:unperturbed time}, the nominal reach time simplifies to 
\begin{equation}\label{eq:T_N^* simplified}
    T_N^*(d) = \underset{\bar{u}_c\, \in\, \bar{U}_c}{\min} \big\{ T \geq 0 : \bar{B}\bar{u}_c\, T = d \big\}.
\end{equation}
The multiplication of the variables $\bar{u}_c$ and $T$ prevents the use of linear solvers.
Instead, we can numerically solve $T_N^*(d) = 1/\underset{\bar{u}_c\, \in\, \bar{U}_c}{\max} \big\{ \lambda : \bar{B} \bar{u}_c = \lambda d \big\}$, after using the transformation $\lambda = \frac{1}{T}$ in \eqref{eq:T_N^* simplified}.

The work \citep{SIAM_CT} showed that $T_N^*$ is absolutely homogeneous when the input set $\bar{U}$ is symmetric. However, in this work the allowable controls \eqref{eq:U bar} are not symmetric and thus $T_N^*$ loses its absolute symmetry but conserves a nonnegative homogeneity.

\begin{proposition}\label{prop:proportional nominal reach time}
    The nominal reach time $T_N^*$ is a nonnegatively homogeneous function of $d$, i.e., $T_N^*(\lambda d) = \lambda \ T_N^*(d)$ for $d \in \mathbb{R}^n$ and $\lambda \geq 0$.
\end{proposition}

\begin{proof}
    Let $d \in \mathbb{R}^n$, $\lambda \geq 0$. The case $\lambda = 0$ is trivial since $T_N^*(0) = 0$, so consider $\lambda > 0$.
    The nominal reach time for $d$ is $T_N^*(d)$, so there exists $\bar{u}_d \in \bar{U}_c$ such that $\bar{B} \bar{u}_d T_N^*(d) = d$. Then, $\bar{B}\, \bar{u}_d\, \lambda T_N^*(d) = \lambda d$. The optimality of $T_N^*(\lambda d)$ to reach $\lambda d$ leads to $T_N^*(\lambda d) \leq \lambda T_N^*(d)$.
    
    There exists $\bar{u}_{\lambda d} \in \bar{U}_c$ such that $\bar{B} \bar{u}_{\lambda d} T_N^*(\lambda d) = \lambda d$, so $\bar{B}\, \bar{u}_{\lambda d}\, \frac{T_N^*(\lambda d)}{\lambda} = d$. The optimality of $T_N^*(d)$ to reach $d$ yields $T_N^*(d) \leq \frac{T_N^*(\lambda d)}{\lambda}$.
    Thus, $\lambda T_N^*(d) \leq T_N^*(\lambda d)$.    
\end{proof}

We can now tackle the dynamics of the malfunctioning system after a loss of control authority over some of its actuators.

\section{Dynamics of the Malfunctioning System}\label{section:malfunctioning dynamics}

We study the system of dynamics \eqref{eq:splitted order 1}, with the aim of computing the malfunctioning reach time $T_M^*$. 
We define the constant input sets $U_c := \big\{ u \in \mathbb{R}^m : u_i \in [u_i^{min}, u_i^{max}]\ \text{for}\ i \in [\![1,m]\!] \big\}$, $W_c := \big\{ w \in \mathbb{R}^p : w_i \in [w_i^{min}, w_i^{max}]\ \text{for}\ i \in [\![1,p]\!] \big\}$, and $V_c$ the set of vertices of $W_c$.

\begin{definition}\label{def:vertex}
    A \emph{vertex} of a set $X \subset \mathbb{R}^n$ is a point $x \in X$ such that if there are $x_1 \in X$, $x_2 \in X$ and $\lambda \in [0,1]$ with $x = \lambda x_1 + (1-\lambda)x_2$, then $x = x_1 = x_2$. 
\end{definition}

\begin{proposition}\label{prop:u cst}
    For a resilient system \eqref{eq:splitted order 1}, $d \in \mathbb{R}_*^n$ and $w \in W$, a constant control input $u_d^*(w) \in U_c$ achieves the infimum $T_M(w,d)$ of \eqref{eq:malfunctioning reach time} defined as
    \begin{equation}\label{eq:T_M}
        T_M(w,d) := \underset{u\, \in\, U}{\inf} \Big\{ T \geq 0 : \int_0^T \hspace{-2mm} Bu(t) + Cw(t)\, dt = d \Big\}.
    \end{equation}
\end{proposition}
\begin{proof}
    Let $d \in \mathbb{R}_*^n$, $w \in W$ and $T_M(w,d) = \underset{u\, \in\, U}{\inf} \big\{ T \geq  0 : \int_0^T Bu(t)\, dt = z \big\}$, with $z := d - \int_0^T Cw(t)\, dt \in \mathbb{R}^n$ a constant vector because $w$ is fixed. Since the system is resilient, any $z \in \mathbb{R}^n$ is reachable. Additionally, $U$ is convex and compact, and $\int_0^T Bu(t)\, dt$ is linear in $u$. Then, according to Theorem 4.3 of \citep{liberzon} a time-optimal control exists. Following the proof of Proposition~\ref{prop:unperturbed time}, we conclude that the infimum of \eqref{eq:T_M} is a minimum and that the optimal control $u_d^*(w)$ belongs to $U_c$.
\end{proof}

We can now work on the supremum of \eqref{eq:malfunctioning reach time}.

\begin{proposition}\label{prop:w cst}
     For a resilient system \eqref{eq:splitted order 1} and $d \in \mathbb{R}_*^n$, the supremum $T_M^*(d)$ of \eqref{eq:malfunctioning reach time} is achieved with a constant undesirable input $w^* \in W_c$.
\end{proposition}
\begin{proof}
    For $d \in \mathbb{R}_*^n$ and with $u_d^*$ defined in Proposition~\ref{prop:u cst}, \eqref{eq:malfunctioning reach time} simplifies to $T_M^*(d) = \underset{w\, \in\, W}{\sup} \big\{ T :  B u_d^*(w) T + \int_0^T Cw(t)\, dt = d \big\}$. Let $w^c := \int_0^{T_M(w,d)} \frac{w(t)}{T_M(w,d)}\, dt$ for $w \in W$. Then, for all $t \geq 0$ and $i \in [\![1,p]\!]$ we have $w_i^{min} \leq w_i(t) \leq w_i^{max}$. Thus, 
    \begin{equation*}
         \int_0^{T_M(w,d)} \hspace{-2mm} \frac{w_i^{min}\, dt}{T_M(w,d)} = w_i^{min} \leq \int_0^{T_M(w,d)} \hspace{-2mm} \frac{w_i(t)\, dt}{T_M(w,d)} = w^c_i \leq \int_0^{T_M(w,d)} \hspace{-2mm} \frac{w_i^{max}\, dt}{T_M(w,d)} = w_i^{max}.
    \end{equation*}
    So, $w^c \in W_c$. Then, $\int_0^{T_M(w,d)} Cw(t)\, dt = C w^c T_M(w,d) = d - B u_d^*(w) T_M(w,d)$. Conversely, note that for all $w_c \in W_c$ and $T > 0$, we can define $w(t) := \frac{1}{T}w_c$ for $t \in [0,T]$ such that $\int_0^T Cw(t)\, dt = C w_c$ and $w \in W$.
    Therefore, the constraint space of the supremum of \eqref{eq:malfunctioning reach time} can be restricted to $W_c$.
    
    \vspace{1mm}
    
    We define the function $\varphi : W_c \rightarrow \mathbb{R}^n$ as $\varphi(w_c) := B u_d^*(w_c) + C w_c$ for $w_c \in W_c$.
    When applying $w_c$ and $u_d^*(w_c)$ the dynamics become $\dot x = \varphi(w_c)$.
    Neustadt in \citep{Neustadt} introduces $\mathcal{A}_{W_c} := \big\{ (x_1, T) : \ \text{for}\ w_c \in W_c, \int_0^T \varphi(w_c)\, dt = x_1 - x_0 \big\}$ as the attainable set from $x_0$ and using inputs in $W_c$.
    Since $\big(Bu_d^*(w_c) + Cw_c \big) T_M(w_c, d) = d$, we have $\varphi(w_c) = \frac{1}{T_M(w_c, d)}d$ and $\varphi$ is continuous in $w_c$ according to Lemma~\ref{lemma: T continuous}.
    Set $W_c$ is compact and $x_0 \in \mathbb{R}^n$ is fixed. Then, Theorem~1 of \citep{Neustadt} states that $\mathcal{A}_{W_c}$ is compact.
    
    Note that $T_M^*(d) = \sup \big\{ T : (x_{goal}, T) \in \mathcal{A}_{W_c} \big\}$, then $T_M^*(d)$ is the supremum of a continuous function over the compact set $\mathcal{A}_{W_c}$, so the supremum of \eqref{eq:malfunctioning reach time} is a maximum achieved on $W_c$. 
\end{proof}

Following Propositions \ref{prop:u cst} and \ref{prop:w cst}, the malfunctioning reach time becomes
\begin{equation}\label{eq:T_M^* with W_c and U_c}
    T_M^*(d) = \max_{w_c\, \in\, W_c} \left\{ \underset{u_c\, \in\, U_c}{\min} \big\{ T \geq 0 : \big(Bu_c + Cw_c\big) T = d \big\} \right\}.
\end{equation}
The simplifications achieved so far were based on existence theorems from \citep{liberzon, Neustadt} upon which the bang-bang principle relies. The logical next step is to show that the maximum of \eqref{eq:T_M^* with W_c and U_c} is achieved by the extreme undesirable inputs, i.e., at the set of vertices of $W_c$, which we denote by $V_c$. However, most of the work on the bang-bang principle considers systems with a linear dependency on the input \citep{LaSalle, liberzon, Sussmann}, while $\varphi$ introduced in Proposition~\ref{prop:w cst} is not linear in the input $w_c$. 

The work from Neustadt \citep{Neustadt} considers a nonlinear $\varphi$, yet his discussion on bang-bang inputs would require us to show that $co(\varphi(W_c)) = co(\varphi(V_c))$. Since $\varphi$ is not linear, such a task is not trivial and in fact it amounts to proving that inputs in $V_c$ can do as much as inputs in $W_c$, i.e., we would need to prove the bang-bang principle.

Two more works \citep{Aronsson, Glashoff} consider bang-bang properties for systems with nonlinear dependency on the input. However, both of them require conditions that are not satisfied in our case. Work contained in \citep{Aronsson} needs the subsystem $\dot x = Cw$ to be controllable, while \citep{Glashoff} requires $T_M$ defined in \eqref{eq:T_M} to be concave in $w_c$.
Thus, even if bang-bang theory seems like a natural approach to restrict the constraint space from $W_c$ to $V_c$ in \eqref{eq:T_M^* with W_c and U_c}, we need a new optimization result, namely Theorem~2.1 from \citep{Maxmax_Minimax_Quotient_thm}. To employ this result, we first need to relate resilience to an inclusion of polytopes.
\begin{definition}\label{def:polytopes}
    A \emph{polytope} in $\mathbb{R}^n$ is a compact intersection of finitely many half-spaces.
\end{definition}

We define the sets $X := \big\{ Cw_c : w_c \in W_c\big\}$ and $Y := \big\{ Bu_c : u_c \in U_c \big\}$.

\begin{proposition}\label{prop:resilience and polytopes}
    For a system following \eqref{eq:splitted order 1}, $X$ and $Y$ are polytopes in $\mathbb{R}^n$. If the system is resilient, then $\dim Y = n$ and $-X \subseteq Y^\circ$.
\end{proposition}
\begin{proof}
    Sets $U_c$ and $W_c$ are defined as polytopes in $\mathbb{R}^m$ and $\mathbb{R}^p$ respectively. Sets $X$ and $Y$ are linear images of $W_c$ and $U_c$, so they are polytopes in $\mathbb{R}^n$ \citep{inf_dim_analysis}.
    
    For a resilient system, following Propositions~\ref{prop:u cst} and \ref{prop:w cst} we know that for all $w_c \in W_c$ and all $d_0 \in \mathbb{R}^n$ there exists $u_c \in U_c$ and $T \geq 0$ such that $(Bu_c + Cw_c)T = d_0$. It also means that for all $x \in X$ and all $d_0 \in \mathbb{R}^n$ there exists $y \in Y$ and $T \geq 0$ such that $(x + y)T = d_0$.
    Since $d_0$ can be freely chosen in $\mathbb{R}^n$, we must have $\dim Y = n$.
    
    Take $x \in X$, $x \neq 0$ and $d_0 = x$. Then, there exists $y_1 \in Y$ and $T_1 > 0$ such that $(x+y_1)T_1 = x$. Hence, $\lambda_1 x \in Y$ with $\lambda_1 := -1 + 1/T_1$.
    Now take $d_0 = -x$. Then, there exists $y_2 \in Y$ and $T_2 > 0$ such that $(x+y_2)T_2 = -x$. Hence, $\lambda_2 x \in Y$ with $\lambda_2 := -1 - 1/T_2$.
    Since $\lambda_2 \leq -1 \leq \lambda_1$ and $Y$ is convex, we have $-x \in Y$.
    
    If $x = 0$, this process fails because we would get $T = 0$ when taking $d = 0$. Instead, take $d_0 \in \mathbb{S}$, then there exists $T > 0$ and $y \in Y$ such that $yT = d_0$. Repeating this for $-d_0$ and using the convexity of $Y$, we obtain $0 \in Y$. Thus $-X \subset Y$.
    
    Now assume that there exists $-x_1 \in -X \cap \partial Y$. Take $d = -x_1$, then the best input is $y = -x_1 \in \partial Y$ because $-X \subset Y$. Then, $x_1 + y = 0 \notin \mathbb{R}_*^+ d$, which contradicts the resilience of the system. Therefore, $-X \cap \partial Y = \emptyset$, i.e., $-X \subset Y^\circ$.
\end{proof}

We can now prove that the maximum of \eqref{eq:T_M^* with W_c and U_c} is achieved on $V_c$.

\begin{proposition}\label{prop:w on a vertex}
    For a resilient system \eqref{eq:splitted order 1} and $d \in \mathbb{R}_*^n$, the maximum of \eqref{eq:T_M^* with W_c and U_c} is achieved with a constant input $w^* \in V_c$.
\end{proposition}
\begin{proof}
    Using $\lambda = \frac{1}{T}$ in \eqref{eq:T_M^* with W_c and U_c} yields $T_M^*(d) = 1/\underset{x\, \in\, X}{\min} \big\{ \underset{y\, \in\, Y}{\max} \big\{ \lambda \geq 0 : x + y = \lambda d \big\} \big\}$.
    Since $\lambda \geq 0$, we can write $\lambda = |\lambda| = \| \lambda d\| / \|d\| = \|x + y\| / \|d\|$. Then, 
     \begin{equation}\label{eq:T_M^*(d) eq with x,y}
        T_M^*(d) = \frac{\|d\|}{\underset{x\, \in\, X}{\min} \big\{ \underset{y\, \in\, Y}{\max} \big\{ \|x + y\| : x+y \in \mathbb{R}^+ d \big\} \big\}}.
    \end{equation}
    Following Proposition~\ref{prop:resilience and polytopes}, sets $X$ and $Y$ are polytopes in $\mathbb{R}^n$, $-X \subset Y^\circ$ and $\dim Y = n$.
    Then, we can apply Theorem~2.1 of \citep{Maxmax_Minimax_Quotient_thm} and conclude that the minimum $x^*$ of \eqref{eq:T_M^*(d) eq with x,y} must be realized on a vertex of $X$. Now, we want to show that $x^* \in \big\{Cv : v \in V_c\big\}$.
    
    Let $w_c \in W_c$ such that $x^* = Cw_c$. If $w_c \in V_c$ we are done. 
    Otherwise, two possibilities remain.
    In the first case $w_c$ is on the boundary of the hypercube $W_c$ and then we take $F$ to be the surface of lowest dimension of $\partial W_c$ such that $w_c \in F$ and $\dim F \geq 1$.
    The other possibility is that $w_c \in W_c^\circ$; we then define $F := W_c$.
    Thus, in both cases $V_c \cap F \neq \emptyset$ and $F$ is convex.
    Then, we take $v \in V_c \cap F$ and $a := v - w_c \in F$. Since $\dim F \geq 1$ and $w_c \in F$, there exists some $\alpha > 0$ such that $w_c \pm \alpha a \in F$. Then
    \begin{equation*}
        x^* = Cw_c = C \Big(\frac{w_c + \alpha a}{2} + \frac{w_c - \alpha a}{2} \Big) = \frac{1}{2} x_+ + \frac{1}{2}x_-,
    \end{equation*}
    with $x_\pm := C(w_c \pm \alpha a)$. Since $x^*$ is a vertex of $X$ and $x_\pm \in X$, according to Definition~\ref{def:vertex}, $x^* = x_+ = x_-$. Then, $x^* - x_- = \alpha Ca = 0$, which yields $Ca = 0$ because $\alpha > 0$. Thus, $Cv = C(w_c + a) = Cw_c = x^*$ and $v \in V_c$.
    Therefore, the maximum of \eqref{eq:T_M^* with W_c and U_c} is achieved on $V_c$.  
\end{proof}

We have reduced the constraint set of \eqref{eq:malfunctioning reach time} from an infinite-dimensional set $W$ to a finite set $V_c$ of cardinality $2^p$, with $p$ being the number of malfunctioning actuators.
Following Propositions \ref{prop:u cst}, \ref{prop:w cst} and \ref{prop:w on a vertex}, the malfunctioning reach time can now be calculated with
\begin{equation}\label{eq:optimization problem}
    T_M^*(d) = \max_{w_c\, \in\, V_c} \left\{ \underset{u_c\, \in\, U_c}{\min} \big\{ T \geq 0 : \big(Bu_c + Cw_c\big) T = d \big\} \right\}.
\end{equation}

It is logic to wonder if the minimum of \eqref{eq:optimization problem} could be restricted to the vertices of $U_c$, just like we did for the maximum over $W_c$. However, that is not possible. Indeed, $w_c$ is chosen freely in $W_c$ in order to make $T_M^*$ as large as possible, while $u_c$ is chosen to counteract $w_c$ and make $Bu_c + Cw_c$ collinear with $d$. This constraint could not be fulfilled for all $d \in \mathbb{R}^n$ if $u_c$ was only chosen among the vertices of $U_c$.

Similarly to the nominal reach time, $T_M^*$ is also linear in the target distance.

\begin{proposition}\label{prop:proportional malfunctioning reach time}
    The malfunctioning reach time $T_M^*$ is a nonnegatively homogeneous function of $d$, i.e., $T_M^*(\lambda d) = \lambda \ T_M^*(d)$ for $d \in \mathbb{R}^n$ and $\lambda \geq 0$.
\end{proposition}

\begin{proof}
    Because of the minimax structure of \eqref{eq:optimization problem}, scaling like in the proof of Proposition~\ref{prop:proportional nominal reach time} is not sufficient to prove the homogeneity of $T_M^*(d)$.
    According to Remark~\ref{rmk:d=0}, for $d = 0$ we have $T_M^*(d) = 0$, so $T_M^*$ is absolutely homogeneous at $d = 0$.
    
    Let $d \in \mathbb{R}_*^n$, $w_c \in W_c$, $x = Cw_c$ and $y^*(x, d) := \arg \underset{y\, \in\, Y}{\min} \big\{ T \geq 0 : (y+x)T = d \big\}$.
    Note that $Bu_d^*(w_c) + Cw_c = y^*(x,d) + x$, with $u_d^*$ defined in Proposition~\ref{prop:u cst}. Then, with $T_M$ introduced in \eqref{eq:T_M}, we have $\big(Bu_d^*(w_c) + Cw_c\big) T_M(w_c, d) = d$, i.e., $y^*(x,d) = \frac{1}{T_M(w_c, d)}d - x$. For $\lambda > 0$, we define $\alpha(\lambda) := \frac{\lambda}{T_M(w_c, \lambda d)} - \frac{1}{T_M(w_c, d)}$, such that $y^*(x, \lambda d) - y^*(x,d) = \alpha(\lambda) d$.
    
    The polytope $Y$ of $\mathbb{R}^n$ has a finite number of faces, so we can choose $d \in \mathbb{R}_*^n$ not collinear with any face of $Y$. Since $Y$ is convex, the ray $\big\{ y^*(x,d) + \alpha d : \alpha \in \mathbb{R} \big\}$ intersects with $\partial Y$ at most twice. 
    Since $y^*(x,d) \in \partial Y$, one intersection happens at $\alpha = 0$. If there exists another intersection, it occurs for some $\alpha_0 \neq 0$.
    Since $y^*(x, \lambda d) \in \partial Y$, we have $y^*(x,d) + \alpha(\lambda) d \in \partial Y$. Then, $\alpha(\lambda) \in \{0, \alpha_0\}$ for all $\lambda > 0$.
    
    According to Lemma~\ref{lemma: T continuous}, $T_M$ is continuous in $d$, so $\alpha$ is continuous in $\lambda$ but its codomain is finite. Therefore, $\alpha$ is constant and we know that $\alpha(1) = 0$. So $\alpha$ is null for all $\lambda > 0$, leading to $T_M(w_c, \lambda d) = \lambda T_M(w_c, d)$ for $\lambda > 0$ and $d$ not collinear with any face of $\partial Y$. Since the dimension of the faces of $\partial Y$ is at most $n-1$ in $\mathbb{R}^n$ and $T_M$ is continuous in $d$, the homogeneity of $T_M$ holds on the whole of $\mathbb{R}^n$.
    Note that $T_M^*(d) = \underset{w_c\, \in\, W_c}{\max} T_M(w_c, d)$. Thus, $\lambda T_M^*(d) = T_M^*(\lambda d)$ for $\lambda > 0$ and $d \in \mathbb{R}^n$.
\end{proof}

We can now combine the initial and malfunctioning dynamics in order to evaluate the quantitative resilience of the system.

\section{Quantitative Resilience}\label{section:r_q}

Quantitative resilience is defined in \eqref{eq:r_q} as the infimum of $T_N^*(d) / T_M^*(d)$ over $d \in \mathbb{R}^n$. Using Proposition~\ref{prop:proportional nominal reach time} and Proposition~\ref{prop:proportional malfunctioning reach time} we reduce this constraint to $d \in \mathbb{S}$.
Focusing on the loss of control over a single actuator we will simplify tremendously the computation of $r_q$. 
In this setting, we can determine the optimal $d \in \mathbb{S}$ by noting that the effects of the undesirable inputs are the strongest along the direction described by the malfunctioning actuator. This intuition is formalized below.

\begin{theorem}\label{thm:direction maximizing t}
    For a resilient system following \eqref{eq:splitted order 1} with $C$ a single column matrix, the direction $d$ maximizing the ratio of reach times $t(d)$ is collinear with the direction $C$, i.e., $\underset{d\, \in\, \mathbb{S}}{\max}\ t(d) = \max\big\{t(C), t(-C) \big\}$.
\end{theorem}
\begin{proof}
    Let $d \in \mathbb{S}$. We use the same process that yielded \eqref{eq:T_M^*(d) eq with x,y} in Proposition~\ref{prop:w on a vertex} but we start from \eqref{eq:T_N^* simplified} where we split $\bar{B}$ into $B$ and $C$:
    \begin{align}\label{eq:T_N^* with x, y}
        \frac{1}{T_N^*(d)} &= \underset{\bar{u}\, \in\, \bar{U}_c}{\max} \big\{ \lambda : \bar{B}\bar{u} = \lambda d\big\} = \underset{u_c\, \in\, U_c,\, w_c\, \in\, W_c}{\max} \big\{ \lambda : Bu_c + Cw_c = \lambda d \big\} \\
        &= \underset{x\, \in\, X,\, y\, \in\, Y}{\max} \big\{ \|y+x\| : y + x \in \mathbb{R}^+ d \big\}. \nonumber
    \end{align}
    We can now gather \eqref{eq:T_M^*(d) eq with x,y} with $d \in \mathbb{S}$ and \eqref{eq:T_N^* with x, y} into
    \begin{equation*}
        t(d) = \frac{T_M^*(d)}{T_N^*(d)} = \frac{\underset{x\, \in\, X,\, y\, \in\, Y}{\max} \big\{ \|x + y\| : x + y \in \mathbb{R}^+ d\big\}}{\underset{x\, \in\, X}{\min} \big\{ \underset{y\, \in\, Y}{\max} \big\{ \|x + y\| : x + y \in \mathbb{R}^+ d\big\} \big\} }.
    \end{equation*}
    Since the system is resilient, Proposition~\ref{prop:resilience and polytopes} states that $X$ and $Y$ are polytopes in $\mathbb{R}^n$, $-X \subset Y^\circ$ and $\dim Y = n$. Because $C$ is a single column, $\dim X = 1$. Then, the Maximax Minimax Quotient Theorem of \citep{Maxmax_Minimax_Quotient_thm} states that $\underset{d\, \in\, \mathbb{S}}{\max}\ t(d) = \max\big\{ t(C), t(-C) \big\}$.     
\end{proof}

Since the sets $U_c$ and $W_c$ are not symmetric, $t$ is not an even function.
Thus, to calculate the quantitative resilience $r_q$ we need to evaluate $T_N^*(\pm C)$ and $T_M^*(\pm C)$, i.e., solve four optimization problems. 
The computation load can be halved with the following result.

\begin{theorem}\label{thm:computation of r_q}
    For a resilient system losing control over a single nonzero column $C$, $r_q = \min \big\{r(C), r(-C)\big\}$, where
    \begin{equation*}
        r(C) := \frac{w^{min} + \lambda^+}{w^{max} + \lambda^+},\quad r(-C) := \frac{w^{max} - \lambda^-}{w^{min} - \lambda^-},\quad \text{with} \ \ \lambda^\pm := \underset{\upsilon \, \in\, U_c}{\max} \big\{ \lambda : B \upsilon = \pm \lambda C \big\}.
    \end{equation*}
\end{theorem}
\begin{proof}
    Let $\bar{u} \in \bar{U}_c$, $u \in U_c$ and $w \in W_c$ be the arguments of the optimization problems \eqref{eq:T_N^* simplified} and \eqref{eq:optimization problem} for $d = C \neq 0$. 
    We split $\bar{u} = [u_B\ u_C]^\top$ such that $u_B \in U_c$ and $u_C \in W_c$. Then,
    \begin{equation}\label{eq:T_N^*(C) and T_M^*(C)}
        \bar{B} \bar{u}\, T_N^*(C) = Bu_B\, T_N^*(C) + Cu_C\, T_N^*(C) = C\ \ \text{and}\ \ Bu\, T_M^*(C) + Cw\, T_M^*(C) = C.
    \end{equation}
    We consider the loss of a single actuator, thus $W_c = [w^{min}, w^{max}] \subset \mathbb{R}$ which makes $Cw T_M^*(C)$ and $Cu_C T_N^*(C)$ collinear with $C$. From Proposition~\ref{prop:w on a vertex}, we know that $w \in \partial W$. Since $w$ maximizes $T_M^*(C)$ in \eqref{eq:T_N^*(C) and T_M^*(C)}, we obviously have $w = w^{min}$. On the contrary, $u_C$ is chosen to minimize $T_N^*(C)$ in \eqref{eq:T_N^*(C) and T_M^*(C)}, so $u_C = w^{max}$.
    
    According to \eqref{eq:T_N^*(C) and T_M^*(C)}, $Bu_B$ and $Bu$ are then also collinear with $C$. The control inputs $u_B$ and $u$ are chosen to minimize respectively $T_N^*(C)$ and $T_M^*(C)$ in \eqref{eq:T_N^*(C) and T_M^*(C)}. Therefore, they are both solutions of the same optimization problem:
    \begin{equation*}\label{}
        \tau^+ = \underset{\upsilon \, \in\, U_c}{\min} \big\{ \tau : B \upsilon \tau = C \big\} \qquad \text{with} \qquad u = u_B = \arg \underset{\upsilon \, \in\, U_c}{\min} \big\{ \tau : B \upsilon \tau = C \big\}.
    \end{equation*}
    We transform this problem into a linear one using the transformation $\lambda = \frac{1}{\tau}$:
    \begin{equation*}
        \lambda^+ = \underset{\upsilon \, \in\, U_c}{\max} \big\{ \lambda : B \upsilon = \lambda C \big\} \qquad \text{with} \qquad u = u_B = \arg \underset{\upsilon \, \in\, U_c}{\max} \big\{ \lambda : B \upsilon = \lambda C \big\}.
    \end{equation*}
    By combining all the results, \eqref{eq:T_N^*(C) and T_M^*(C)} simplifies into:
    \begin{equation*}
        C(\lambda^+ + w^{max}) T_N^*(C) = C\ \ \text{and}\ \ C(\lambda^+ + w^{min}\big) T_M^*(C) = C.\ \ \text{So,}\ \frac{T_N^*(C)}{T_M^*(C)} = \frac{\lambda^+ + w^{min}}{\lambda^+ + w^{max}}.
    \end{equation*}
    Following the same reasoning for $d = -C$, we obtain
    \begin{equation*}
        C(-\lambda^- + w^{min}) T_N^*(C) = -C \quad \text{and} \quad C(-\lambda^- + w^{max}\big) T_M^*(C) = -C,
    \end{equation*}
    with $\lambda^- = \underset{\upsilon\, \in\, U_c}{\max} \big\{ \lambda : B \upsilon = - \lambda C\big\}$. Then, $\frac{1}{t(-C)} = \frac{T_N^*(-C)}{T_M^*(-C)} = \frac{w^{max} - \lambda^-}{w^{min} - \lambda^-}$.
    Following Theorem~\ref{thm:direction maximizing t}, $r_q = \frac{1}{\max\{ t(d)\, :\, d\, \in\, \mathbb{S}\}} = \min \left\{\frac{1}{t(C)}, \frac{1}{t(-C)} \right\} = \min \big\{ r(C), r(-C)\big\} = r_{min}$.
\end{proof}

We introduced quantitative resilience as the solution of four nonlinear nested optimization problems and with Theorem~\ref{thm:computation of r_q} we reduced $r_q$ to the solution of two linear optimization problem.
We can then quickly calculate the maximal delay caused by the loss of control of a given actuator.

\section{Resilience and Quantitative Resilience}\label{section:resilience conditions}

So far, all our results need the system to be resilient. However, based on \citep{journal_paper} verifying the resilience of a system is not an easy task. Besides, as explained in Section~\ref{section:preliminaries}, the resilience criteria from \citep{journal_paper} do not apply here because the set of admissible controls are different. Proposition~\ref{prop:resilient full rank} is only a necessary condition for resilience, while we are looking for an equivalence condition.

\begin{proposition}\label{prop:resilience = full rank + finite T_M^*(C)}
    A system following \eqref{eq:system order 1} is resilient to the loss of control over a column $C$ if and only if it is controllable and $\max \big\{ T_M^*(C), T_M^*(-C)\big\}$ is finite.
\end{proposition}
\begin{proof}
    First, assume that the system \eqref{eq:system order 1} is resilient. Then, according to Proposition~\ref{prop:resilient full rank} for $k = 1$, the system $\dot x(t) = Bu(t)$ is controllable. Since $Im(B) \subset Im(\bar{B})$, the system \eqref{eq:system order 1} is controllable a fortiori. If $C \neq 0$, then following Proposition~\ref{prop:w cst}, $T_M^*(C)$ and $T_M^*(-C)$ are finite. If $C = 0$, then $T_M^*(C)$ is also finite according to Remark~\ref{rmk:d=0}.
    
    \vspace{2mm}
    Now, assume that the system \eqref{eq:system order 1} is controllable and $\max \big\{T_M^*(C), T_M^*(-C)\big\}$ is finite. Let $w \in W_c$ and $d \in \mathbb{R}_*^n$.
    By controllability, there exists $\bar{u} \in \bar{U}_c$ and $\lambda > 0$ such that $\bar{B} \bar{u} = \lambda d$. We split $\bar{B}$ into $[B\ C]$ and same for $\bar{u}$ into $(u_d, w_d)$. Then, $u_d \in U_c$ and $\bar{B} \bar{u} = Bu_d + Cw_d = \lambda d$.
    
    In the case $C = 0$, this yields $Bu_d = \lambda d = Bu_d + Cw$, so the system is resilient.
    
    For $C \neq 0$, we will first show that for any $w \in W_c$ we can find $u \in U_c$ such that $Bu + Cw = 0$.
    Because $T_M^*(C)$ and $T_M^*(-C)$ are finite, $T_M(w, \pm C)$ is positive and finite for all $w \in W_c = [w^{min}, w^{max}]$, with $T_M(\cdot,\cdot)$ defined in \eqref{eq:T_M}. Take $w \in W_c$. Then, there exists $u_+^w \in U_c$ and $u_-^w \in U_c$ such that $\big( B u_+^w + Cw \big) T_M(w,C) = C$ and $\big( B u_-^w + Cw \big) T_M(w,-C) = -C$.
    Define $\alpha := \frac{T_M(w,C)}{T_M(w,C) + T_M(w,-C)} \in (0,1)$ and $u := \alpha u_+^w + (1-\alpha) u_-^w \in U_c$ because $U_c$ is convex. Then, notice that
    \begin{align*}
        Bu + Cw &= \alpha\big( Bu_+^w + Cw\big) + (1-\alpha) \big( Bu_-^w + Cw\big) \\
        &\hspace{-10mm}= \frac{T_M(w,C)}{T_M(w,C) + T_M(w,-C)} \frac{C}{T_M(w,C)} + \frac{T_M(w,-C)}{T_M(w,C) + T_M(w,-C)} \frac{-C}{T_M(w,-C)} = 0.
    \end{align*}
    We want to make a convex combination of $u$ and $u_d$ to build the desired control, but without an extra step that will not work if $w \in \partial W_c$. We first need to show that even if $w$ is a little bit outside of $W_c$ we can still counteract it. Let $\varepsilon := \min \left( \frac{1}{2T_M(w^{min}, C)}, \frac{1}{2T_M(w^{max},-C)} \right) > 0$. Now take $w' \in (w^{max}, w^{max} + \varepsilon]$. There exists $u_- \in U_c$ and $u_+ \in U_c$ such that $\big( B u_+ + Cw^{max} \big) T_M(w^{max},C) = C$ and $\big( B u_- + Cw^{max} \big) T_M(w^{max},-C) = -C$.
    Then, we can define $T^+ > 0$ such that
    \begin{align*}
        B u_+ \hspace{-1mm} + Cw' \hspace{-1mm} = B u_+ \hspace{-1mm} + Cw^{max} \hspace{-1mm} + C(w' - w^{max}) = C \hspace{-1mm} \left( \hspace{-1mm} \frac{1}{T_M(w^{max},C)} + w' - w^{max} \hspace{-1mm} \right) \hspace{-1mm} = \frac{C}{T^+}.
    \end{align*}
    Since $w' - w^{max} \leq 1/2T_M(w^{max},-C)$, we can similarly define $T^- > 0$ such that
    \begin{align*}
        B u_- + Cw' = -C \left( \frac{1}{T_M(w^{max},-C)} - (w' - w^{max}) \right) = -\frac{C}{T^-}.
    \end{align*}
    Similar to above we take $\alpha = \frac{T^+}{T^+ + T^-} \in (0,1)$ making $u' = \alpha u_+ + (1-\alpha) u_- \in U_c$ by convexity and $Bu' + Cw' = 0$. An analogous approach holds for $w' \in [w^{min} - \varepsilon, w^{min})$.
    
    Since $W_c$ is convex, $w \in W_c$ and $w_d \in W_c$, we can take $w' \in [w^{min} - \varepsilon, w^{max} + \varepsilon]$ such that there exists $\gamma \in (0,1)$ for which $w = \gamma w_d + (1-\gamma)w'$.
    We build $u' \in U_c$ as above to make $Bu' + Cw' = 0$. By convexity of $U_c$, $u := \gamma u_d + (1-\gamma)u' \in U_c$. Then,
    \begin{equation*}
        Bu + Cw = \gamma \big(Bu_d + Cw_d\big) + (1-\gamma) \big(Bu' + Cw'\big) = \gamma \lambda d + 0.
    \end{equation*}
    Since $\gamma > 0$, we have $\gamma \lambda > 0$ making the system resilient to the loss of column $C$.  
\end{proof}

The intuition behind Proposition~\ref{prop:resilience = full rank + finite T_M^*(C)} is that a resilient system must fulfill two conditions: being able to reach any state, this is controllability, and doing so in finite time despite the worst undesirable inputs, which corresponds to $T_M^*(\pm C)$ being finite.

Our goal is to relate resilience and quantitative resilience through the value of $r_{min}$. To breach the gap between this desired result and Proposition~\ref{prop:resilience = full rank + finite T_M^*(C)}, we evaluate the requirements on the ratio $\frac{T_N^*(\pm C)}{T_M^*(\pm C)}$ for a system to be resilient.

\begin{corollary}\label{cor:resilience T_N(C)/T_M(C)}
    A system following \eqref{eq:system order 1} is resilient to the loss of control over a column $C$ if and only if it is controllable, $\frac{T_N^*(C)}{T_M^*(C)} \in (0, 1]$ and $\frac{T_N^*(-C)}{T_M^*(-C)} \in (0, 1]$.
\end{corollary}
\begin{proof}
    First, assume that the system \eqref{eq:system order 1} is resilient. Then, according to Proposition~\ref{prop:resilient full rank}, it is controllable. For $d \in \mathbb{R}_*^n$, following Propositions~\ref{prop:unperturbed time} and \ref{prop:w cst}, $T_N^*(d)$ and $T_M^*(d)$ are both finite and positive. Using the separation $\bar{B} = [B\ C]$ and $\bar{u} = [u_B\ u_C]$,
    \begin{align*}
        T_N^*(d) = \underset{\bar{u}\, \in\, \bar{U}_c}{\min} \big\{ T \geq 0 : \bar{B}\bar{u}T = d\big\} &= \underset{w\, \in\, W_c,\, u\, \in\, U_c}{\min} \big\{ T \geq 0 : (Bu + Cw)T = d \big\} \\
        & \hspace{-5mm}\leq \underset{w\, \in\, W_c}{\max} \big\{ \underset{u\, \in\, U_c}{\min} \big\{ T \geq 0 : (Bu + Cw)T = d \big\} \big\} = T_M^*(d).
    \end{align*}
    If $C \in \mathbb{R}_*^n$, we then have $0 < \frac{T_N^*(C)}{T_M^*(C)} \leq 1$ and $0 < \frac{T_N^*(-C)}{T_M^*(-C)} \leq 1$. If $C = 0$, following Remark~\ref{rmk:d=0} we have $\frac{T_N^*(0)}{T_M^*(0)} = 1$.
    
    \vspace{2mm}
    
    Now, assume that the system is controllable, $\frac{T_N^*(C)}{T_M^*(C)} \in (0, 1]$ and $\frac{T_N^*(-C)}{T_M^*(-C)} \in (0, 1]$. 
    If $C = 0$, then $T_M^*(C) = 0$ according to Remark~\ref{rmk:d=0}. We conclude with Proposition~\ref{prop:resilience = full rank + finite T_M^*(C)} that the system is resilient.
    
    Now for the case where $C \neq 0$, let $d \in \mathbb{R}_*^n$.
    Since the system following \eqref{eq:system order 1} is controllable, $T_N^*(\pm C)$ is finite. Since $C \neq 0$, we have $T_N^*(\pm C) > 0$. If $T_M^*(C) = +\infty$, then $\frac{T_N^*(C)}{T_M^*(C)} = 0$, which contradicts the assumption. By definition, $T_M^*(C) \geq 0$, thus $T_M^*(C)$ is finite. The same holds for $T_M^*(-C)$. Then, according to Proposition~\ref{prop:resilience = full rank + finite T_M^*(C)}, the system is resilient.  
\end{proof}

Theorem~\ref{thm:computation of r_q} allows us to compute $r_q$ for resilient systems with a linear optimization.
We now want to extend that result to non-resilient systems, by showing that $r_{min}$ also indicates whether the system is resilient.

\begin{corollary}\label{cor:resilience lambda}
    A system following \eqref{eq:system order 1} is resilient to the loss of control over a nonzero column $C$ if and only if it is controllable, and $r(C)$, $r(-C)$ from Theorem~\ref{thm:computation of r_q} are in $(0,1]$.
\end{corollary}
\begin{proof}
    For a resilient system with $C \neq 0$, following Theorem~\ref{thm:computation of r_q} $r_q = r_{min}$. Since $r_{min} = \min \left\{ \frac{T_N^*(C)}{T_M^*(C)}, \frac{T_N^*(-C)}{T_M^*(-C)} \right\}$, according to Corollary~\ref{cor:resilience T_N(C)/T_M(C)} the resilient system is controllable and $r_q \in (0,1]$.
    
    \vspace{2mm}
    
    Now assume that the system is controllable, $\frac{w^{min} + \lambda^+}{w^{max} + \lambda^+}$ and $\frac{w^{max} - \lambda^-}{w^{min} - \lambda^-} \in (0, 1]$.
    If $w^{min} + \lambda^+ < 0$, then $w^{max} + \lambda^+ \leq w^{min} + \lambda^+$ because $r(C) \in (0,1]$. This leads to the impossible conclusion that $w^{max} \leq w^{min}$. If $w^{min} + \lambda^+ = 0$, then $r(C) = 0$. Therefore, $w^{min} + \lambda^+ > 0$. Let $u \in U_c$ such that $Bu = \lambda^+ C$. For $w \in W_c$, we define $T_w := \frac{1}{w + \lambda^+}$, so that $(Bu + Cw)T_w = C$. Note that $T_w$ is positive and finite because $w + \lambda^+ \geq w^{min} + \lambda^+ > 0$.
    Notice that $T_M^*(C) \leq \underset{w\, \in\, W_c}{\max} T_w = \frac{1}{w^{min} + \lambda^+}$. And since $T_M^*(C) \geq 0$, it is finite. 
    
    The same reasoning holds for $r(-C)$. We can show that $w^{max} - \lambda^- < 0$ and that $T_w := \frac{1}{\lambda^- - w} > 0$ for all $w \in W_c$. With $u \in U_c$ such that $Bu = -\lambda^- C$ we have $(Bu+Cw)T_w = -C$. Then, $T_M^*(-C) \leq \underset{w\, \in\, W_c}{\max} T_w = \frac{1}{\lambda^- - w^{max}}$, so $T_M^*(-C)$ is finite.
    Then, Proposition~\ref{prop:resilience = full rank + finite T_M^*(C)} states that the system is resilient.
\end{proof}

We now have all the tools to assess the quantitative resilience of a driftless system. If $\bar{B}$ is not full rank, the system following \eqref{eq:system order 1} is not controllable and thus not resilient. Otherwise, we compute the ratios $r(\pm C)$ and Corollary~\ref{cor:resilience lambda} states whether the system is resilient. If it is, then $r_q = r_{min}$ by Theorem~\ref{thm:computation of r_q}, otherwise $r_q = 0$. We summarize this process in Algorithm~\ref{algo}.

\RestyleAlgo{ruled}

\begin{algorithm}[hbpt!]
\setstretch{1.2}
\caption{Resilience algorithm for system \eqref{eq:system order 1}}\label{algo}
\KwData{A column $C$ of $\bar{B}$, $r(C)$ and $r(-C)$ from Theorem~\ref{thm:computation of r_q}}
\KwResult{$r_q$}
\eIf{$\rank(\bar{B}) = n$ and $0 \in \interior(\bar{\mathcal{U}})$}{
    \hspace{27mm} \# {\fontfamily{cmtt}\selectfont \small system \eqref{eq:system order 1} is controllable} \break
    \eIf{$r(C) \in (0, 1]$ and $r(-C) \in (0, 1]$}{
        $r_q = \min\{ r(C), r(-C)\}$ \# {\fontfamily{cmtt}\selectfont \small resilient to loss of $C$}
    }{
    $r_q = 0$  \hspace{10mm} \# {\fontfamily{cmtt}\selectfont \small not resilient to loss of $C$}
    }
    }{
    $r_q = 0$ \hspace{17mm} \# {\fontfamily{cmtt}\selectfont \small not resilient to any loss}
}
\end{algorithm}

\section{Systems with Multiple Integrators}\label{section:integrators}

We can now extend the results obtained for driftless systems to generalized higher-order integrators.

\begin{proposition}\label{prop:T_k,N^*(d)}
    If system \eqref{eq:system order 1} is controllable, then for all $k \in \mathbb{N}$ system \eqref{eq:system order k} is controllable. The infimum of \eqref{eq:T_k,N^*} is achieved with the same constant control input $\bar{u}^* \in \bar{U}_c$ as $T_N^*$ in \eqref{eq:nominal reach time}. Additionally, $T_{k,N}^*(d) = \sqrt[\leftroot{-2}\uproot{2}k]{k!\ T_N^*(d)}$ for all $d \in \mathbb{R}^n$. 
\end{proposition}
\begin{proof}
    Let $k \in \mathbb{N}$, $x_{goal} \in \mathbb{R}^n$ and $d := x_{goal} - x_0$. If $d = 0$, then $T_{k,N}^*(d) = 0 = T_N^*(d)$, so the result holds.
    Let $d \neq 0$. By assumption, system $\dot y(t) = \bar{B}\bar{u}(t)$ with $y(0) = 0$ is controllable. Following Proposition~\ref{prop:unperturbed time} there exists a constant optimal control $\bar{u} \in \bar{U}_c$ such that $y\big(T_N^*(d)\big) - y(0) = d = \bar{B}\bar{u} T_N^*(d)$, with $T_N^*(d) > 0$. Then, applying the control input $\bar{u}$ to \eqref{eq:system order k} on the time interval $[0,\, t_1]$ leads to
    \begin{align*}
        x(t_1) - x_0 &= \hspace{-1mm} \int_0^{t_1} \hspace{-3mm}\int_0^{t_2} \hspace{-3mm} \hdots \int_0^{t_k} \hspace{-2mm} x^{(k)}(t_{k+1})\, dt_{k+1} \hdots dt_2 = \hspace{-1mm} \int_0^{t_1} \hspace{-3mm} \int_0^{t_2} \hspace{-3mm} \hdots \int_0^{t_k} \hspace{-2mm} \bar{B}\bar{u}\, dt_{k+1} \hdots dt_2 \\
        &= \bar{B} \bar{u} \frac{t_1^k}{k!} = \frac{d}{T_N^*(d)} \frac{t_1^k}{k!},
    \end{align*}
    since $x^{(l)}(0) = 0$ for $l \in [\![1, k-1]\!]$ and $\bar{B} \bar{u} = \frac{d}{T_N^*(d)} \in \mathbb{R}^n$ is constant. 
    By taking $t_1 = \sqrt[\leftroot{-2}\uproot{2}k]{k!\ T_N^*(d)}$, we obtain $x(t_1) - x_0 = d$.
    Thus, the state $x_{goal}$ is reachable in finite time $t_1$, so the system \eqref{eq:system order k} is controllable and $T_{k,N}^*(d) \leq t_1$.
    
    \vspace{2mm}
    
    Assume for contradiction purposes that there exists $\tilde{u} \in \bar{U}$ such that the state of \eqref{eq:system order k} can reach $x_{goal}$ in a time $\tau \in (0, t_1)$. Since $\tilde{u}$ can be time-varying, we build $\hat{u} := \frac{k!}{\tau^k} \int_0^\tau \hdots \int_0^{t_k} \tilde{u}(t_{k+1})\, dt_{k+1} \hdots dt_2$, constant vector of $\mathbb{R}^{m+p}$. Since $\tilde{u} \in \bar{U}$, $\tilde{u}_i(t) \in [\bar{u}_i^{min}, \bar{u}_i^{max}]$ for all $i \in [\![1, m+p]\!]$ and $t \in [0,\tau]$. Because $\bar{u}_i^{min}$ and $\bar{u}_i^{max}$ are constant, one can easily obtain through $k$ successive integrations that $\hat{u}_i \in [\bar{u}_i^{min}, \bar{u}_i^{max}]$ for all $i \in [\![1, m+p]\!]$. And thus, $\hat{u} \in \bar{U}_c$, i.e., $\hat{u}$ is an admissible constant control input.
    Then, we apply $\tilde{u}$ to \eqref{eq:system order k} on the time interval $[0, \tau]$ and we obtain
    \begin{equation*}
        x(\tau) - x_0 = d = \hspace{-1mm} \int_0^{\tau} \hspace{-3mm} \hdots \hspace{-1mm} \int_0^{t_k} \hspace{-3mm} x^{(k)}(t_{k+1})\, dt_{k+1} \hdots dt_2 = \hspace{-1mm} \int_0^{\tau} \hspace{-3mm} \hdots \hspace{-1mm} \int_0^{t_k} \hspace{-3mm} \bar{B}\tilde{u}(t_{k+1})\, dt_{k+1} \hdots dt_2 = \bar{B} \hat{u} \frac{\tau^k}{k!},
    \end{equation*}
    so $\bar{B} \hat{u} = \frac{k!d}{\tau^k}$.
    Applying the control input $\hat{u}$ to the system $\dot y(t) = \bar{B}\bar{u}(t)$ on the interval $[0,T]$ with $T := \frac{\tau^k}{k!}$ leads to
    \begin{equation*}
        y(T) = \int_0^{T} \hspace{-2mm} \dot y(t)\, dt = \int_0^{T} \hspace{-2mm} \bar{B} \hat{u}\, dt = \bar{B} \hat{u} T = \frac{k!d}{\tau^k} \frac{\tau^k}{k!} = d. 
    \end{equation*}
    Thus, $y$ can reach $d$ in a time $T = \frac{\tau^k}{k!} < \frac{t_1^k}{k!} = T_N^*(d)$, which contradicts the optimality of $T_N^*(d)$.
    Then, $t_1$ is the minimal time for the state of \eqref{eq:system order k} to reach $x_{goal}$. Therefore, the infimum of \eqref{eq:T_k,N^*} is achieved with the same constant input $\bar{u} \in \bar{U}_c$ as $T_N^*(d)$ in \eqref{eq:nominal reach time}, and $T_{k,N}^*(d) = \sqrt[\leftroot{-2}\uproot{2}k]{k!\ T_N^*(d)}$.
\end{proof}

\begin{remark}\label{rmk: non-zero_initial_conditions}
    The proof of Proposition~\ref{prop:T_k,N^*(d)} went smoothly because the initial condition had zero derivatives. We will study a simple case with non-zero initial condition and show that even the existence of $T_{k,N}^*$ is not obvious.

    Let $k = 2$ and denote $v := \dot x$. Assume that system $\dot v = \bar{B}\bar{u}$ is controllable, and $\dot x(0) = v(0) = v_0 \neq 0$. For any $v_1 \in \mathbb{R}_*^n$ there exists an optimal control $\bar{u} \in \bar{U}_c$ such that $v\big(T_N^*(d_1)\big) - v_0 = v_1 = \bar{B}\bar{u} T_N^*(v_1)$, with $T_N^*(v_1) > 0$, since $v_1 \neq 0$.
    Then,
    \begin{align*}
        v(t) - v_0 &= \int_0^t \dot v(\tau)\, d\tau = \int_0^t \bar{B}\bar{u}\, d\tau = \int_0^t \frac{v_1}{T_N^*(v_1)}\, d\tau = \frac{t v_1}{T_N^*(v_1)}, \quad \text{and} \\
        x(T) - x_0 &= \int_0^T v(t)\, dt = \int_0^T \frac{t v_1}{T_N^*(v_1)} + v_0\, dt = \frac{T^2 v_1}{2T_N^*(v_1)} + v_0T.
    \end{align*}
    Taking $x(T) = x_{goal}$ leads to a quadratic polynomial in $\mathbb{R}^n$: $\frac{v_1}{2T_N^*(v_1)} T^2 + v_0 T - d = 0$. These are $n$ scalar equations for $n+1$ unknowns: $v_1$ and $T$. Because $T_N^*(v_1)$ depends on $v_1$, the equations are not independent and thus might not have a solution.
    Then, even for this seemingly simple case, the existence of $T_{2,N}^*$ is not obvious to justify.
\end{remark}
A result similar to Proposition~\ref{prop:T_k,N^*(d)} holds for the malfunctioning reach time of order $k$.

\begin{proposition}\label{prop:T_k,M^*(d)}
    If system \eqref{eq:splitted order 1} is resilient, then system \eqref{eq:splitted order k} is resilient for all $k \in \mathbb{N}$. The supremum and infimum of \eqref{eq:T_k,M^*} are achieved with the same constant inputs $u^* \in U_c$ and $w^* \in W_c$ as $T_M^*$ in \eqref{eq:malfunctioning reach time}, and $T_{k,M}^*(d) = \sqrt[\leftroot{-2}\uproot{2}k]{k!\ T_M^*(d)}$ for $d \in \mathbb{R}^n$.
\end{proposition}
\begin{proof}
    Let $k \in \mathbb{N}$, $x_{goal} \in \mathbb{R}^n$ and $d := x_{goal} - x_0$. If $d = 0$, then $T_{k,M}^*(d) = 0 = T_M^*(d)$, so the result holds.
    From now on we assume that $d \neq 0$. As shown in the proof of Proposition~\ref{prop:T_k,N^*(d)} we can work with only constant inputs.
    First, we need to prove that the function $u_d^* : W_c \rightarrow U_c$ defined in Proposition~\ref{prop:u cst} produces the best control input $u_d^*(w)$ for any undesirable input $w \in W_c$.

    By assumption, system $\dot y(t) = Bu(t) + Cw(t)$ with $y(0) = 0$ is resilient. Let $w \in W_c$ an undesirable input. Then, $y\big(T_M(w,d)\big) - y(0) = d = \big( Bu_d^*(w) + Cw \big) T_M(w,d)$, with $T_M(w,d) > 0$ defined in Proposition~\ref{prop:u cst}. Then, applying $u_d^*(w)$ and $w$ to \eqref{eq:splitted order k} on the time interval $[0, t_w]$ leads to
    \begin{align*}
        x(t_w) - x_0 &= \hspace{-1mm} \int_0^{t_w} \hspace{-3mm}\int_0^{t_2} \hspace{-3mm} \hdots \int_0^{t_k} \hspace{-3mm} x^{(k)}(t_{k+1})\, dt_{k+1} \hdots dt_2 = \hspace{-1mm} \int_0^{t_w} \hspace{-3mm} \hdots \int_0^{t_k} \hspace{-3mm} Bu_d^*(w)+Cw\, dt_{k+1} \hdots dt_2 \\
        &= \big( Bu_d^*(w)+Cw \big) \frac{t_w^k}{k!} = \frac{d}{T_M(w,d)} \frac{t_w^k}{k!},
    \end{align*}
    since $x^{(l)}(0) = 0$ for $l \in [\![1, k-1]\!]$ and $Bu_d^*(w)+Cw = \frac{d}{T_M(w,d)} \in \mathbb{R}^n$ is constant. 
    By taking $t_w = \sqrt[\leftroot{-2}\uproot{2}k]{k!\ T_M(w,d)}$, we obtain $x(t_w) - x_0 = d$.
    
    \vspace{2mm}
    
    Assume for contradiction purposes that for the same $w$ there exists $u \in U_c$ such that the state of \eqref{eq:splitted order k} reaches $x_{goal}$ at a time $\tau \in (0, t_w)$. Then,
    \begin{align*}
        x(\tau) - x_0 &= d = \hspace{-1mm} \int_0^{\tau} \hspace{-3mm} \hdots \hspace{-1mm} \int_0^{t_k} \hspace{-3mm} x^{(k)}(t_{k+1})\, dt_{k+1} \hdots dt_2 = \hspace{-1mm} \int_0^{\tau} \hspace{-3mm} \hdots \hspace{-1mm} \int_0^{t_k} \hspace{-3mm} Bu+Cw\, dt_{k+1} \hdots dt_2 \\
        &= (Bu+Cw) \frac{\tau^k}{k!},
    \end{align*}
    so $(Bu+Cw) = \frac{k!d}{\tau^k}$.
    Applying the control $u$ and undesirable input $w$ to the system $\dot y(t) = Bu(t) + Cw(t)$ on the time interval $[0,T]$ with $T := \frac{\tau^k}{k!}$ leads to
    \begin{equation*}
        y(T) = \int_0^{T} \hspace{-2mm} \dot y(t)\, dt = \int_0^{T} \hspace{-2mm} Bu+Cw\, dt = (Bu + Cw) T = \frac{k!d}{\tau^k} \frac{\tau^k}{k!} = d. 
    \end{equation*}
    Thus, $y$ reaches $d$ in a time $T = \frac{\tau^k}{k!} < \frac{t_w^k}{k!} = T_M(w,d)$, which contradicts the optimality of $T_M(w,d)$. Therefore, $u_d^*(w)$ is the best control input to counteract any $w \in W_c$ for system \eqref{eq:splitted order k}.

    \vspace{2mm}
    
    Now we need to prove that $w^*$ defined in Proposition~\ref{prop:w cst} is the worst undesirable input for system \eqref{eq:splitted order k}.
    With the control input $u_d^*(w^*) = u^*$, the state of \eqref{eq:splitted order k} verifies $x(t_1) - x_0 = d$ at $t_1 = \sqrt[\leftroot{-2}\uproot{2}k]{k!\ T_M(w^*, d)} = \sqrt[\leftroot{-2}\uproot{2}k]{k!\ T_M^*(d)}$ because $w^*$ is the worst undesirable input for system \eqref{eq:splitted order 1}.
    
    Assume for contradiction purposes that there exists some $w_2 \in W_c$ such that for the control $u_d^*(w_2) \in U_c$ the state of \eqref{eq:splitted order k} can only reach $x_{goal}$ in a time $\tau > t_1$. Then,
    \begin{equation*}
        x(\tau) - x_0 = d = \hspace{-1mm} \int_0^{\tau} \hspace{-3mm} \hdots \hspace{-1mm} \int_0^{t_k} \hspace{-3mm} Bu_d^*(w_2)+Cw_2\, dt \hdots dt_2 = \big( Bu_d^*(w_2)+Cw_2) \frac{\tau^k}{k!},
    \end{equation*}
    so $(Bu_d^*(w_2)+Cw_2) = \frac{k!d}{\tau^k}$.
    Applying the control $u_d^*(w_2)$ and undesirable input $w_2$ to the system $\dot y(t) = Bu(t) + Cw(t)$ on the time interval $[0,T]$ with $T := \frac{\tau^k}{k!}$ leads to
    \begin{equation*}
        y(T) = \int_0^{T} \hspace{-2mm} \dot y(t)\, dt = \int_0^{T} \hspace{-2mm} Bu_d^*(w_2)+Cw_2\, dt = \big( Bu_d^*(w_2) + Cw_2 \big) T = \frac{k!d}{\tau^k} \frac{\tau^k}{k!} = d. 
    \end{equation*}
    Thus, $y$ cannot reach $d$ in a time shorter than $T$. But $T = \frac{\tau^k}{k!} > \frac{t_1^k}{k!} = T_M^*(d)$, which contradicts the optimality of $w^*$ as worst undesirable input for \eqref{eq:splitted order 1}. Then, $w^*$ is also the worst undesirable input for the system \eqref{eq:splitted order k}.
    
    Therefore, the supremum and the infimum of \eqref{eq:T_k,M^*} are achieved with the same constant inputs $w^* \in W_c$ and $u^* \in U_c$ as $T_M^*$ in \eqref{eq:malfunctioning reach time}, and $T_{k,M}^*(d) = \sqrt[\leftroot{-2}\uproot{2}k]{k!\ T_M^*(d)}$.
\end{proof}

Since $T_{k,M}^*$ is related to $T_M^*$ with the same formula as $T_{k,N}^*$ is related to $T_N^*$, we can exploit all previous results established for the system of dynamics \eqref{eq:system order 1}.

\begin{theorem}\label{thm:r_k,q}
    If system \eqref{eq:system order 1} is resilient, then for all $k \in \mathbb{N}$ system \eqref{eq:system order k} is resilient and $r_{k,q} = \sqrt[\leftroot{-1}\uproot{2}k]{r_q}$.
\end{theorem}
\begin{proof}
    Let $d \in \mathbb{R}_*^n$, then based on Propositions~\ref{prop:T_k,N^*(d)} and \ref{prop:T_k,M^*(d)}, \eqref{eq:t_k(d)} becomes
    \begin{equation*}
        t_k(d) = \frac{T_{k,M}^*(d)}{T_{k,N}^*(d)} = \frac{\sqrt[\leftroot{-2}\uproot{2}k]{k!\ T_M^*(d)}}{\sqrt[\leftroot{-2}\uproot{2}k]{k!\ T_N^*(d)}} = \sqrt[\leftroot{-2}\uproot{2}k]{ \frac{T_M^*(d)}{T_N^*(d)}} = \sqrt[\leftroot{-2}\uproot{2}k]{t(d)}.
    \end{equation*}
    Therefore, $r_{k,q} = \sqrt[\leftroot{-2}\uproot{2}k]{r_q}$.
    For a resilient system $r_q \in (0,1]$, so $r_{k,q} \in (0,1]$ and $r_{k,q} \geq r_q$. Thus, adding integrators to a resilient system increases its quantitative resilience.
\end{proof}

Thus, by studying the system $\dot x(t) = \bar{B} \bar{u}(t)$ we can verify the resilience and calculate the quantitative resilience of any system of the form $x^{(k)}(t) = \bar{B}\bar{u}(t)$ for $k \in \mathbb{N}$.
We will now apply our theory to two numerical examples.

\section{Numerical Examples}\label{section:examples}

Our first example considers a linearized model of a low-thrust spacecraft performing orbital maneuvers. We study the resilience of the spacecraft with respect to the loss of control over some thrust frequencies.
Our second example features an octocopter UAV (Unmanned Aerial Vehicle) enduring a loss of control authority over some of its propellers.

\subsection{Linear Quadratic Trajectory Dynamics}

We study a low-thrust spacecraft in orbit around a celestial body. Because of the complexity of nonlinear low-thrust dynamics the work \citep{ko15} established a linear model for the spacecraft dynamics using Fourier thrust acceleration components. Given an initial state and a target state, the model simulates the trajectory of the spacecraft in different orbit maneuvers, such as an orbit raising or a plane change. The states of this linear model are the orbital elements $x := \big(a, \ e, \ i, \ \Omega, \ \omega, \ M \big)$ whose names are listed in Table~\ref{tbl:ITstates}.

Because of the periodic motion of the spacecraft, the thrust acceleration vector $F$ can be expressed in terms of its Fourier coefficients $\alpha$ and $\beta$:
\begin{align*}
    F = F_R \hat{r} + F_W \hat{w} + F_S (\hat{w} \times \hat{r}) \hspace{3mm} \text{with} \hspace{3mm} F_{R,W,S} = \hspace{-1mm} \sum_{k=0}^{\infty} \big(\alpha^{R,W,S}_k \cos{kE}+\beta^{R,W,S}_k \sin{kE} \big),
\end{align*}
where $F_R$ is the radial thrust acceleration, $F_W$ is the circumferential thrust acceleration, $F_S$ is the normal thrust acceleration and $E$ is the eccentric anomaly.
The work \citep{trajectory_dynamics} determined that only 14 Fourier coefficients affect the average trajectory, and we use those coefficients as the input $\bar{u}$: 
\begin{equation*}
    \bar{u} = \left[\arraycolsep=4.4pt\begin{array}{cccccccccccccc}
    \alpha^R_0 &  \alpha^R_1 &  \alpha^R_2 &  \beta^R_1 &  \alpha^S_0 &  \alpha^S_1 &  \alpha^S_2 &  \beta^S_1 &  \beta^S_2 &  \alpha^W_0 &  \alpha^W_1 &  \alpha^W_2 &  \beta^W_1 &  \beta^W_2
    \end{array}\right]^\top .
\end{equation*}
The Fourier coefficients considered in \citep{trajectory_dynamics} are chosen in $\big[ -2.5 \times 10^{-7}, 2.5 \times 10 ^{-7}\big]$, so we can safely assume that for our case the Fourier coefficients all belong to $[-1, 1]$.
Following \citep{ko15}, the state-space form of the system dynamics is $\dot{x} = \bar{B}(x)\bar{u}$.
We calculate $\bar{B}(x)$ in Appendix~\ref{apx:bar B} using the averaged variational equations for the orbital elements given in \citep{trajectory_dynamics}.
We implement the orbit raising scenario presented in \citep{ko15}, with the orbital elements of the initial and target orbits listed in Table~\ref{tbl:ITstates}.
\begin{table}[htbp!] 
    \begin{center}
    \caption{Initial and Target States for Raising Maneuver}
    \label{tbl:ITstates}
    \begin{tabular}{|c|cc|c|c|} \hline
        Name of the Orbital Elements & \multicolumn{2}{|c|}{Parameters} & initial & target \\ \hline
        semi-major axis & $a$\, &[\rm{km}] & 6678 & 7345 \\ \hline
        eccentricity & $e$\, &[\, -\, ] & 0.67 & 0.737 \\ \hline
        inclination & $i$\, &[\rm{degrees}] & 20 & 22 \\ \hline
        longitude of the ascending node & $\Omega$\, &[\rm{degrees}] & 20 & 22 \\ \hline
        argument of perigee & $\omega$\, &[\rm{degrees}] & 20 & 22 \\ \hline
        mean anomaly & $M$\, &[\rm{degrees}] & 20 & 20 \\ \hline
    \end{tabular}
    \end{center}
\end{table}

We approximate $\bar{B}(x)$ as a constant matrix $\bar{B}$ taken at the initial state. The resulting matrix is: $\bar{B} = 10^{-6}\times$
\begin{equation*}
\left[\arraycolsep=2.4pt
  \begin{array}{cccccccccccccc}
        0 & 0 & 0 & 18314 & 40583 & 0 & 0 & 0 & 0 & 0 & 0 & 0 & 0 & 0 \\
        0 & 0 & 0 & 1.1 & -3.4 & 2.3 & -0.4 & 0 & 0 & 0 & 0 & 0 & 0 & 0 \\
        0 & 0 & 0 & 0 & 0 & 0 & 0 & 0 & 0 & -5.2 & 3.8 & -0.9 & -0.7 & 0.2 \\
        0 & 0 &	0 &	0 &	0 &	0 &	0 &	0 &	0 &	-5.5 & 4 & -0.9 & 5.6 & -1.9 \\
        3 &	-2.7 & 0 & 0 & 0 & 0 & 0 & 4.7 & -1 & 5.2 &	-3.8 & 1.3 & -5.6 & 1.9 \\
        -12.3 & 7.2 & -0.9 & 0 & 0 & 0 & 0 & -3.5 & 0.8 & 0 & 0 & 0 & 0 & 0
    \end{array} \right] \hspace{-1mm}.
\end{equation*}
Coefficients $\bar{B}_{1,4}$ and $\bar{B}_{1,5}$ are significantly larger than all the other coefficients of $\bar{B}$ because the semi-major axis is larger than any other element, as can be seen in Table~\ref{tbl:ITstates}. 
Losing control over one of the 14 Fourier coefficients means that a certain frequency of the thrust vector cannot be controlled. 
Since the coefficients $\bar{B}_{1,5}$ and $\bar{B}_{6,1}$ have a magnitude significantly larger than coefficients of respectively the first and last row of $\bar{B}$, we have the intuition that the system is not resilient to the loss of the $1^{st}$ or the $5^{th}$ Fourier coefficient.

The matrix $\bar{B}$ is full rank, so $\dot x = \bar{B}\bar{u}$ is controllable. We denote with $r_{min}$ and $r_q$ the vectors whose components are respectively $r_{min}(j)$ and $r_q(j)$ for the loss of the frequency $j \in [\![1, 14]\!]$, $r_{min} =$
\begin{equation*}
     \left[\arraycolsep=3pt\begin{array}{cccccccccccccc}
    -0.2 & 0.34 & 0.9 & -0.004 & -0.38 & 0.15 & 0.83 & -0.32 & 0.71 & -0.06 & 0.24 & 0.2 & -0.5 & 0.5
    \end{array}\right] \hspace{-1mm}.
\end{equation*}
Since the $1^{st}$, $4^{th}$, $5^{th}$, $8^{th}$, $10^{th}$, and $13^{th}$ values of $r_{min}$ are negative, according to Corollary~\ref{cor:resilience lambda} the system is not resilient to the loss of control over any one of these six corresponding frequencies. Their associated $r_q$ is zero. This result validates our intuition about the $1^{st}$ and $5^{th}$ frequencies.
Corollary~\ref{cor:resilience lambda} also states the resilience of the spacecraft to the loss over any one of the $2^{nd}$, $3^{rd}$, $6^{th}$, $7^{th}$, $9^{th}$, $11^{th}$, $12^{th}$ and $14^{th}$ frequency because their $r_{min}$ belongs to $(0,1]$. Indeed, the input bounds are symmetric, so we can use the results from \citep{SIAM_CT} stating that $r(C) = r(-C) = r_{min}$. Then, using Theorem~\ref{thm:computation of r_q} we deduce that
\begin{align*}
    r_q = \left[\arraycolsep=4.8pt\begin{array}{cccccccccccccc}
    0 & 0.34 & 0.9 & 0 & 0 & 0.15 & 0.83 & 0 & 0.71 & 0 & 0.24 & 0.2 & 0 & 0.5
    \end{array}\right].
\end{align*}

Since $r_q(3)$, $r_q(7)$ and $r_q(9)$ are close to $1$, the loss of one of these three frequency would not delay significantly the system.
The lowest positive value of $r_q$ occurs for the $6^{th}$ frequency, $r_q(6) = 0.15$. Its inverse, $\frac{1}{r_q(6)} = 6.8$ means that the malfunctioning system can take up to 6.8 times longer than the initial system to reach a target.

The maneuver described in Table~\ref{tbl:ITstates} yields $d = x_{goal} - x_0 = \big(667,\, 0.067,\, 2,\, 2,\, 2,\, 2\big) $. We compute the associated time ratios $t(d)$ using \eqref{eq:T_N^* simplified} and \eqref{eq:optimization problem} for the loss over each column of $\bar{B}$:
\begin{equation}\label{eq:t(d) specific target}
  t(d) =  \left[\arraycolsep=3.5pt\begin{array}{cccccccccccccc}
    \hspace{-1mm} 1.1 & 1.2 & 1.1 & 1 & \infty & 1 & 151.1 & \infty & 151.1 & \infty & 151.1 & 151.1 & \infty & 151.1 \hspace{-1mm}
    \end{array}\right] \hspace{-1mm}.
\end{equation}
Then, losing control over one of the first four frequencies will barely increase the time required for the malfunctioning system to reach the target compared with the initial system.
However, after the loss over the $7^{th}$, $9^{th}$, $11^{th}$, $12^{th}$, or the $14^{th}$ frequency of the thrust vector, the undesirable input can multiply the maneuver time by a factor of up to $151.1$.
If one of the $5^{th}$, $8^{th}$, $10^{th}$, or the $13^{th}$ frequency is lost, then some undesirable inputs can render the maneuver impossible to perform.

When computing $r_q$, we have seen that the system is not resilient to the loss of the $1^{st}$ or the $4^{th}$ frequency. Yet, the specific target described in Table \ref{tbl:ITstates} happens to be reachable for the same loss since the $1^{st}$ and $4^{th}$ components of $t(d)$ in \eqref{eq:t(d) specific target} are finite. Indeed, $r_q$ speaks only about a target for which the undesirable inputs cause maximal possible delay.

\subsection{A Resilient Octocopter}

Resilience of unmanned aerial vehicles (UAV) to system failure is crucial to their operations over populated areas \citep{unstable_quadcopters}. The most common design for UAV is the quadcopter with four horizontal propellers.
Quadcopters have six degrees of freedom (position and orientation) but only four inputs: the angular velocities of the propellers. These systems are thus underactuated and cannot be resilient to the loss of control authority over one of their propeller \citep{unstable_quadcopters}.

To remedy this crucial safety concern the solution is to consider overactuated drones like octocopters \citep{UAV_equations}. Most octocopters models have only horizontal propellers as represented on Figure~\ref{fig:flat_octo}, so they must be tilted to operate an horizontal motion, which can be an issue for some payloads.
An innovative solution has been devised in \citep{drone_model}, where four propellers are horizontal and four are vertical, as represented on Figure~\ref{fig:octorotor}. This design decouples the rotational and the translational dynamics, which simplify the control of the UAV. In this section, we evaluate the quantitative resilience of such an octocopter model.
\begin{figure}[htbp!]
    \centering
    \begin{subfigure}{.45\textwidth}
        \centering
        \includegraphics[scale = 0.29]{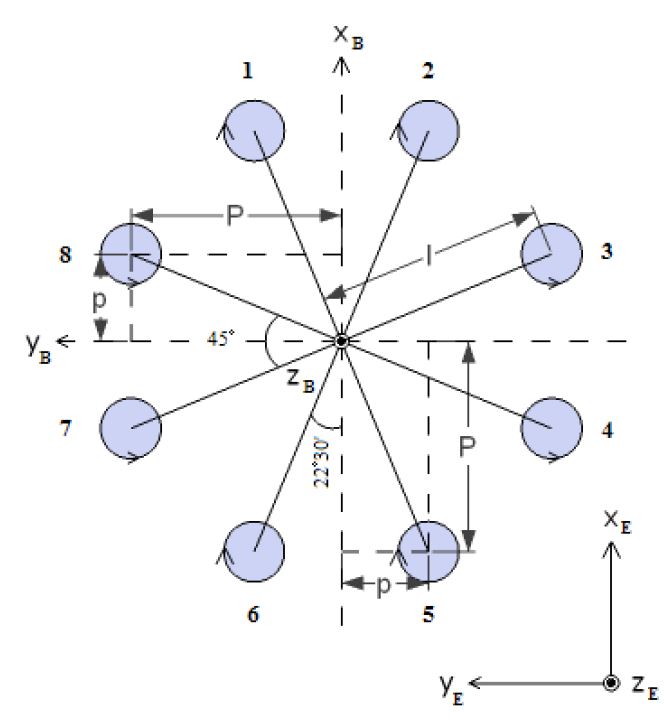}
        \caption{Octocopter layout from \citep{UAV_equations}.}
        \label{fig:flat_octo}
    \end{subfigure}
    \begin{subfigure}{.54\textwidth}
        \centering
        \includegraphics[scale = 0.26]{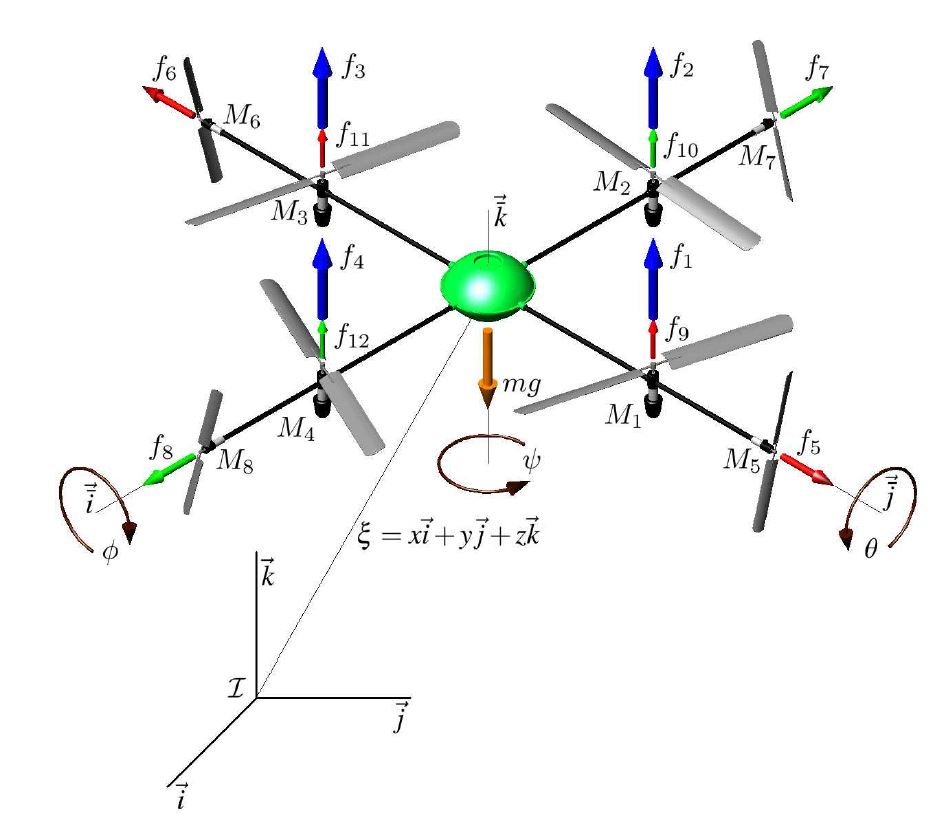}
        \caption{Octocopter layout from \citep{drone_model}.}
        \label{fig:octorotor}
    \end{subfigure}
    \vspace{-5mm}
   \caption{Two possible octocopter configurations.}
\end{figure}
\vspace{-5mm}

\subsubsection{Rotational Dynamics}

The roll, pitch and yaw angles of the octocopter are gathered in $Y := (\phi, \theta, \psi)$.
The propeller $i \in [\![1,8]\!]$ spinning at an angular velocity $\omega_i$ produces a force $f_i = k \omega_i^2$, with $k$ the thrust coefficient. The airflow created by the lateral rotors produces an extra vertical force referred as $f_9$ to $f_{12}$ on Figure~\ref{fig:octorotor}. Then, $f_{9+i} = b f_{5+i}$ for $i \in [\![0,3]\!]$ with the coupling constant $b = 0.64$ from \citep{drone_model}.
Relying on \citep{UAV_equations} and \citep{drone_model}, the rotational dynamics of the octocopter are
\begin{align*}
    \ddot{\phi} &= \frac{I_y - I_z}{I_x} \dot \theta \dot \psi + \frac{lk}{I_x}\Big( \omega_3^2 - \omega_1^2 + b \big(\omega_7^2 - \omega_8^2\big) \Big) - \frac{I_{rotor}}{I_x} \dot \theta \gamma \\
    \ddot{\theta} &= \frac{I_z - I_x}{I_y} \dot \theta \dot \phi + \frac{lk}{I_y}\Big( \omega_2^2 - \omega_4^2 + b \big(\omega_5^2 - \omega_6^2\big) \Big) + \frac{I_{rotor}}{I_y} \dot \phi \gamma \\
    \ddot{\psi} &= \frac{I_x - I_y}{I_z} \dot \theta \dot \phi + \frac{d}{I_z}\Big( \omega_2^2 + \omega_4^2 - \omega_1^2 - \omega_3^2 \Big),
\end{align*}
with $\gamma = \omega_2 + \omega_4 - \omega_1 - \omega_3$. The numerical values used in \citep{UAV_equations} are: $l = 0.4\ m$ the arm length, $m = 1.64\ kg$ the mass, $I_x = I_y = 0.5 I_z = 0.044\ kg\, m^2$ the inertia, $k = 10^{-5}\ N \, s^2$ the thrust coefficient, $d = 0.3\times 10^{-6}\ N\, m\, s^2$ the drag coefficient, $I_{rotor} = 9\times 10^{-5}\ kg\, m^2$ the rotor inertia, and $\omega_{max} = 8000\ rpm = 838\ rad /s$ the maximal angular velocity of the propellers.
The linearized rotational equations are $\ddot Y(t) = \bar{B}_{rot} \Omega(t)$, with $\Omega(t) \in \mathbb{R}^8$ gathering the squared angular velocities of the propellers, i.e.,
\begin{equation}\label{eq:rotational}
    \ddot Y(t) \hspace{-1mm} = \begin{bmatrix} \ddot{\phi}(t) \\ \ddot{\theta}(t) \\ \ddot{\psi}(t) \end{bmatrix} \hspace{-1mm} = \underbrace{\begin{bmatrix} \frac{lk}{I_x} & 0 & 0 \\ 0 &  \frac{lk}{I_y} & 0 \\ 0 & 0 & \frac{d}{I_z} \end{bmatrix} \hspace{-1mm} \begin{bmatrix} -1 & 0 & 1 & 0 & 0 & 0 & b & -b \\
    0 & 1 & 0 & -1 & b & -b & 0 & 0 \\
    -1 & 1 & -1 & 1 & 0 & 0 & 0 & 0 \end{bmatrix}}_{=\ \bar{B}_{rot}} \hspace{-1mm} \begin{bmatrix} \omega_1^2(t) \\ \omega_2^2(t) \\ \omega_3^2(t) \\ \omega_4^2(t) \\ \omega_5^2(t) \\ \omega_6^2(t) \\ \omega_7^2(t) \\ \omega_8^2(t) \end{bmatrix} \hspace{-1mm} = \bar{B}_{rot} \Omega(t).
\end{equation}

\subsubsection{Quantitative Resilience of the Rotational Dynamics}

The matrix $\bar{B}_{rot}$ in \eqref{eq:rotational} has more columns than rows and each output is affected by four different inputs, so we have the intuition that $\bar{B}_{rot}$ is resilient. 
Because the non-zero coefficients of $\bar{B}_{rot}$ have similar magnitudes to one another and are evenly spread in the matrix, we expect the quantitative resilience to be the same for each actuator.

Since the input sets are nonsymmetric: $\bar{u}_i(t) := \omega_i^2(t) \in [0, \omega_{max}^2]$, and the dynamics are given by a double integrator: $\ddot Y(t) = \bar{B}_{rot} \bar{u}(t)$, the theory of \citep{SIAM_CT} cannot deal with this UAV model. 
Using Theorem~\ref{thm:computation of r_q} we calculate the quantitative resilience of the system $\dot v_Y(t) = \bar{B}_{rot} \bar{u}(t)$ with $v_Y(t) := \dot Y(t)$ for the loss of control over each single propeller: 
\begin{equation*}
    r_{min} = \left[\begin{array}{cccccccc} 0.1 & 0.1 & 0.1 & 0.1 & 0.1 & 0.1 & 0.1 & 0.1 \end{array}\right].
\end{equation*}
Based on Corollary~\ref{cor:resilience lambda}, the UAV is resilient to the loss of control over any single propeller in terms of angular velocity and $r_q = r_{min}$. Following Theorem~\ref{thm:r_k,q} we deduce that $\ddot Y(t) = \bar{B}_{rot} \bar{u}(t)$ is also resilient and $r_{2,q} = \sqrt{r_q} = \sqrt{0.1} \approx 0.32$. Then, after the loss of control over any single propeller the UAV might take as much as three times longer to reach a given orientation, while it might be ten times slower to reach a given angular velocity.

\subsubsection{Translational dynamics}

In the inertial frame the position of the UAV is $X := (x,y,z)$ and its orientation yields the rotation matrix $R(\psi, \theta, \phi)$. The translational equations of motion from \citep{drone_model} are $m \ddot X(t) = R\big(\psi(t), \theta(t), \phi(t) \big) \bar{B}_{trans} k\Omega(t) + G$, i.e.,
\begin{equation}\label{eq:translational}
    m \begin{bmatrix} \ddot{x}(t) \\ \ddot{y}(t) \\ \ddot{z}(t) \end{bmatrix} = R\big(\psi(t), \theta(t), \phi(t) \big) \begin{bmatrix} 0 & 0 & 0 & 0 & 1 & -1 & 0 & 0 \\
    0 & 0 & 0 & 0 & 0 & 0 & 1 & -1 \\
    1 & 1 & 1 & 1 & b & b & b & b \end{bmatrix} k \Omega(t) + \begin{bmatrix} 0 \\ 0 \\ -mg \end{bmatrix}.
\end{equation}
Because of the gravitation term $G$, the above dynamics are affine. We combine $G$ with the input $\Omega$ to make the dynamics driftless using $R(\psi, \theta, \phi)^{-1} = R(-\psi, -\theta, -\phi)$,
\begin{align*}
    m \ddot X(t) &= R(\psi, \theta, \phi) \big( \bar{B}_{trans} k\Omega(t) + R(-\psi, -\theta, -\phi) G \big) \\
    &= R(\psi, \theta, \phi) \begin{bmatrix} k(\omega_5^2 - \omega_6^2) - mg(-c_\psi s_\theta c_\phi + s_\psi s_\phi) \\ k(\omega_7^2 - \omega_8^2) - mg( s_\psi s_\theta c_\phi + c_\psi s_\phi) \\ k(\omega_1^2 + \omega_2^2 + \omega_3^2 + \omega_4^2) + bk(\omega_5^2 + \omega_6^2 + \omega_7^2 + \omega_8^2) - mg c_\theta c_\phi \end{bmatrix}.
\end{align*}
Since the rotational dynamics are resilient, after a loss of control over a propeller, we can maintain the UAV into a level mode, i.e., $\theta = \phi = 0$, the roll and pitch angles are null. We will then only translate the UAV when it is in level mode.
The four horizontal propellers are designed to sustain the weight of the drone while the lateral ones are smaller and should mostly be used for lateral displacements. Thus, we define the inputs as $\bar{u}_i(t) := k\omega_i^2(t) - \frac{mg}{4} \in [-\frac{mg}{4}, k \omega_{max}^2 - \frac{mg}{4}]$ for $i \in [\![1,4]\!]$ and $\bar{u}_i(t) := k\omega_i^2(t) \in [0, k \omega_{max}^2]$ for $i \in [\![5,8]\!]$.
Then,
\begin{equation}\label{eq: trans bang-bang}
    \ddot X(t) = \frac{1}{m} R(\psi, 0, 0) \bar{B}_{trans} \bar{u}(t).
\end{equation}
We have transformed the affine dynamics of the UAV into a driftless form with nonsymmetric inputs and we can deal with such a system only because this work extends the theory of \citep{SIAM_CT}.

\subsubsection{Quantitative Resilience of the Translational Dynamics}

After the loss of control authority over a propeller, we split $\frac{1}{m} R(\psi, 0, 0)\bar{B}_{trans}$ into $B$, $C$, and $\bar{u}$ into $u$, $w$ as before. The initial state is the same and the malfunctioning dynamics are
\begin{equation}\label{eq: split trans bang-bang}
    \ddot X(t) = B u(t) + C w(t).
\end{equation}

Matrix $\bar{B}_{trans}$ in \eqref{eq:translational} has more columns than rows and the first four columns are identical so we expect the system to be resilient to the loss of any one of them.
However, only column 5 can counteract column 6 and only column 7 can counteract column 8, and vice-versa. We thus have the intuition that the system is not resilient to the loss of any one of the last four columns.

For the system $\dot v_X(t) = \frac{1}{m} R(\psi, 0, 0) \bar{B}_{trans} \bar{u}(t)$, with $v_X(t) := \dot X(t)$ we calculate $r(C)$ and $r(-C)$ based on Theorem~\ref{thm:computation of r_q} for the loss of control over each single propeller:
\begin{equation*}
    \def\arraystretch{1.2}\begin{array}{cccccccc} 
    r(C) = \big[ 0.7657 & 0.7657 & 0.7657 & 0.7657 & 0 & 0 & 0 & 0 \big] \\
    r(-C) = \big[ 0.5638 & 0.5638 & 0.5638 & 0.5638 & 0 & 0 & 0 & 0 \big]
    \end{array} \hspace{-1mm}.
\end{equation*}
Then, according to Corollary~\ref{cor:resilience lambda} the system is only resilient to the loss of any one of the first four propellers. Following Theorem~\ref{thm:computation of r_q},
\begin{equation*}
    r_q = \min\big\{ r(C), r(-C)\big\} = \big[ \begin{array}{cccccccc} 0.5638 & 0.5638 & 0.5638 & 0.5638 & 0 & 0 & 0 & 0 \end{array} \big] .
\end{equation*}

We notice that the quantitative resilience is not affected by the value of the yaw angle $\psi$.
Taking the inverse of $r_q$ we obtain the factor describing the effects of the undesirable inputs on the worst case performance of the system:
\begin{equation*}
    \frac{1}{r_q} = \left[\begin{array}{cccccccc} 1.7738 & 1.7738 & 1.7738 & 1.7738 & +\infty & +\infty & +\infty & +\infty \end{array}\right].
\end{equation*}
Hence, the loss of a horizontal propeller can increase the time to reach a certain velocity by a factor up to $1.7738$.
Using Theorem~\ref{thm:r_k,q} we can also study the resilience of $\ddot X(t) = \frac{1}{m} R(\psi, 0, 0) \bar{B}_{trans} \bar{u}(t)$, by calculating $r_{2,q} = \sqrt{r_q} = \big[0.75\ 0.75\ 0.75\ 0.75\ 0\ 0\ 0\ 0 \big]$.
Therefore, the translational dynamics of the octocopter are resilient to the loss of control over any single horizontal propeller. However, they are not resilient to the loss of any vertical propeller, as predicted.

We can also pick a direction of motion and evaluate how the loss of each single actuator would impact the change of velocity in this direction. For the impact on the vertical velocity we take $d = (0,0,-1)$ and
\begin{equation*}
    t(d) = \left[ \begin{array}{cccccccc} 1.7738 & 1.7738 & 1.7738 & 1.7738 & 2.2644 & 2.2644 & 2.2644 & 2.2644 \end{array}\right].
\end{equation*}
Note that the first four values are the same as in $1/r_q$ because the direction that is the worst impacted by a loss of an horizontal propeller is the vertical direction.

If we look at how a change of forward velocity is impacted by a loss of control we take $d = (1,0,0)$, $\psi = 0$ and we obtain $t(d) = \big[1\ 1\ 1\ 1\ +\infty\ +\infty\ 1\ 1 \big]$.
Thus, the four horizontal propellers have no impact on the forward velocity as expected. Losing control over one of the two lateral vertical propellers (columns 7 and 8) does not affect the forward motion.
However, the loss of the front or back vertical propeller (columns 5 and 6) completely prevents a guaranteed forward motion.

\subsubsection{High-fidelity dynamics of the propellers}\label{subsec:actuator dynamics}

So far in this work, all inputs were bang-bang because our definition of quantitative resilience asks for time-optimal transfers. The inputs of the translational dynamics \eqref{eq: trans bang-bang} encode the propellers' angular velocities, which cannot physically change in a bang-bang fashion. Thus, in order to provide a more realistic model and display the capabilities of our work, we follow \citep{ADMIRE_1} and add first-order propellers' dynamics:
\begin{equation}\label{eq: trans smooth}
     \ddot X(t) = \bar{B}_{trans} \bar{u}(t), \quad \dot{\bar{u}}(t) = \frac{1}{\tau}\big(\bar{u}^c(t) - \bar{u}(t) \big),
\end{equation}
with $\bar{u}^c \in \mathbb{R}^8$ a new, possibly bang-bang, command signal. System \eqref{eq: trans smooth} is not driftless, hence preventing a direct application of our theory. Instead, we proceed heuristically, building on the intuition provided by our theory to tackle this high-fidelity model.

The time constant $\tau = 0.1\, s$ is chosen to match the propeller response in Fig.~3 of \citep{UAV_prop}.
After the loss of control over the first propeller, we split $\bar{B}_{trans}$ and $\bar{u}$ as before such that
\begin{equation}\label{eq: split trans smooth}
    \ddot X(t) = B u(t) + C w(t), \hspace{3mm} \left\{ \hspace{-2mm} \def\arraystretch{1.2}\begin{array}{l}
    \dot u(t) = \frac{1}{\tau} \big( u^c(t) - u(t) \big), \\
    \dot w(t) = \frac{1}{\tau} \big( w^c(t) - w(t) \big), \end{array} \right.
\end{equation}
with the bang-bang command signals $u^c$ and $w^c$. We will now study how the actuators' dynamics impact the resilience of the UAV in the vertical direction $d = (0,0,1)$.

\begin{figure}[htbp!]
    \centering
    \includegraphics[scale=0.34]{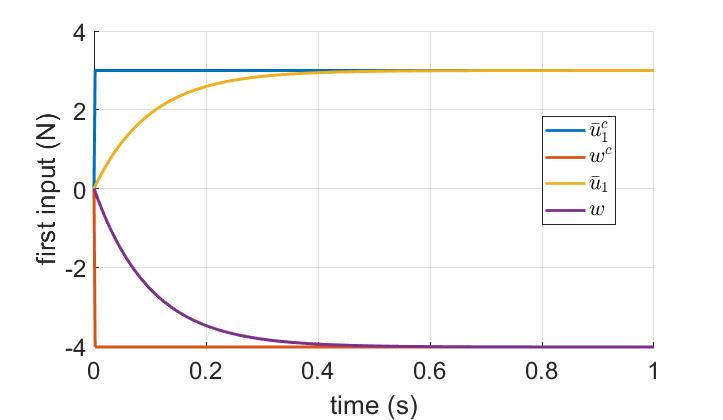}
    \caption{Exponential convergence of $\bar{u}_1$ and $w$ to their bang-bang commands $\bar{u}^c_1 \hspace{-0.3mm} = \hspace{-0.3mm} \bar{u}_1^{max} \hspace{-0.3mm} = \hspace{-0.3mm} k\omega_{max}^2 \hspace{-0.3mm} - \hspace{-0.3mm} \frac{mg}{4}$ and $w^c \hspace{-0.3mm} = \hspace{-0.3mm} \bar{u}_1^{min} \hspace{-0.3mm} = \hspace{-0.3mm} -\frac{mg}{4}$.}
    \label{fig:u_1}
\end{figure}

Since the inputs $\bar{u}$ in \eqref{eq: trans smooth} and $(u,w)$ in \eqref{eq: split trans smooth} have a non-zero rise time as shown on Fig.~\ref{fig:u_1}, the vertical velocities $\dot z_N$ of \eqref{eq: trans smooth} and $\dot z_M$ of \eqref{eq: split trans smooth} react smoothly and slower than their bang-bang counterparts, as illustrated on Fig.~\ref{fig:z_dot}. For $t \geq 0.4\, s$, $\bar{u}$ and $(u,w)$ have converged to their commands $\bar{u}^c$ and $(u^c, w^c)$, and thus the two slopes of $\dot z_N(t)$ in \eqref{eq: trans bang-bang} and \eqref{eq: trans smooth} are equal, as shown on Fig.~\ref{fig:z_dot}, and so are that of $\dot z_M(t)$ in \eqref{eq: split trans bang-bang} and \eqref{eq: split trans smooth}.

\begin{figure}[htbp!]
    \centering
    \includegraphics[scale=0.34]{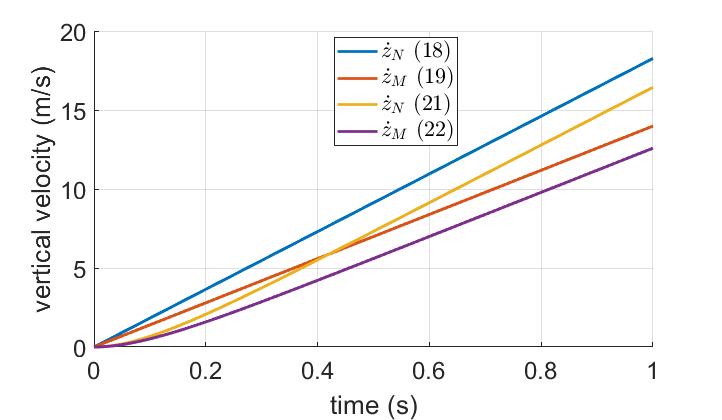}
    \caption{Vertical velocities $\dot z_N(t)$ and $\dot z_M(t)$ of the nominal and malfunctioning systems demonstrating the impact of the propellers' dynamics in \eqref{eq: trans smooth} and \eqref{eq: split trans smooth}.}
    \label{fig:z_dot}
\end{figure}

The slower reaction time caused by the dynamics of the propellers is also reflected on the vertical positions $z_N$ and $z_M$ on Fig.~\ref{fig:z}.

\begin{figure}[htbp!]
    \centering
    \includegraphics[scale=0.34]{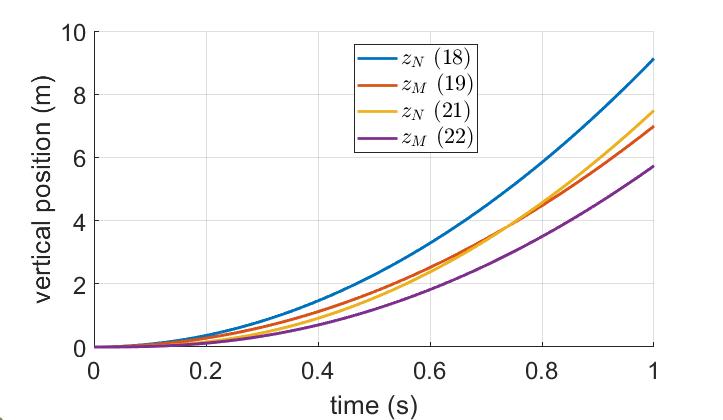}
    \caption{Vertical positions $z_N(t)$ and $z_M(t)$ of the nominal and malfunctioning systems demonstrating the impact of the propellers' dynamics in \eqref{eq: trans smooth} and \eqref{eq: split trans smooth}.}
    \label{fig:z}
\end{figure}

Because of the specific geometry of the system, the optimal inputs for direction $d = (0,0,1)$ were trivial to determine. Then, we calculate the ratio of reach times for systems \eqref{eq: trans smooth} and \eqref{eq: split trans smooth}, $\frac{T_M^*(d)}{T_N^*(d)} = 1.12$ and for systems \eqref{eq: trans bang-bang} and \eqref{eq: split trans bang-bang}, $\frac{T_M^{c*}(d)}{T_N^{c*}(d)} = 1.14$. Hence, modeling the dynamics of the propellers increases slightly the resilience of the vertical dynamics.

 However, the time-optimal commands $\bar{u}^c$ for \eqref{eq: trans smooth} and $(u^c, w^c)$ for \eqref{eq: split trans smooth} can be time-varying for other directions $d \in \mathbb{R}^3$ \citep{liberzon}, and determining these optimal commands requires complex algorithms \citep{Eaton, Sakawa} because the dynamics are not driftless anymore. Additionally, the Maximax-Minimax Quotient Theorem of \citep{Maxmax_Minimax_Quotient_thm} does not hold, which invalidates Theorem~\ref{thm:direction maximizing t} and prevents the exact calculation of $r_q$ without calculating $\frac{T_M^*(d)}{T_N^*(d)}$ for all $d \in \mathbb{R}^3$. A stronger theory will be needed to tackle linear non-driftless systems.

\section{Conclusion and Future Work}

This paper built on the notion of quantitative resilience for control systems introduced in previous work and extended it to linear systems with multiple integrators and nonsymmetric input sets.
Relying on bang-bang control theory and on two novel optimization results, we transformed a nonlinear problem consisting of four nested optimizations into a single linear problem.
This simplification leads to a computationally efficient algorithm to verify the resilience and calculate the quantitative resilience of driftless systems with integrators.

There are two promising avenues of future work. 
Because of the complexity of the subject, we have only considered driftless systems so far. However, future work should be able to extend the concept of quantitative resilience to non-driftless linear systems.
Finally, noting that Theorems~\ref{thm:direction maximizing t} and \ref{thm:computation of r_q} only concern the loss of a single actuator, our second direction of work is to extend these results to the simultaneous loss of multiple actuators.

\appendix

\section{Continuity of a minimum}\label{apx:continuity}

\begin{lemma}\label{lemma: T continuous}
    For a resilient system following \eqref{eq:splitted order 1}, the function 
    $T_M(w_c, d) := \underset{u_c\, \in\, U_c}{\min} \big\{ T \geq 0 : (Bu_c + Cw_c)T = d\big\}$ is continuous in $w_c \in W_c$ and $d \in \mathbb{R}_*^n$.
\end{lemma}
\begin{proof}
    We define sets $X := \big\{Cw_c : w_c \in W_c\big\}$ and $Y := \big\{ Bu_c : u_c \in U_c \big\}$. Then, by abuse of notation $T_M(x, d) = \underset{y\, \in\, Y}{\min} \big\{ T \geq 0 : (y + x)T = d\big\}$. Using the transformation $\lambda = 1/T$, we obtain $T_M(x, d) = 1 / \underset{y\, \in\, Y}{\max} \big\{ \lambda \geq 0 : x + y = \lambda d \big\}$. Since $\|d\| > 0$ and $\lambda \geq 0$, we have $\lambda = \| \lambda d \|/ \|d\| = \|x+y\| / \|d\|$. Let $d_1 := d / \|d\|$ so that $d_1 \in \mathbb{S}$ and $\mathbb{R}^+ d_1 = \mathbb{R}^+ d$. Thus $T_M(x,d) = \|d\| / \underset{y\, \in\, Y}{\max} \big\{ \|x+y\| : x + y \in \mathbb{R}^+ d_1 \big\}$.
    According to Proposition~\ref{prop:resilience and polytopes}, $X$ and $Y$ are polytopes in $\mathbb{R}^n$ and $-X \subset Y$. 
    Then, Lemma~5.2 of \citep{Maxmax_Minimax_Quotient_thm} states that $T_M$ is continuous in $w_c$ and $d$.   
\end{proof}

\section{Equation of Motion for the Low-Thrust Spacecraft}\label{apx:bar B}

The control matrix $\bar{B}$ can be written as
\begin{equation*}
    \bar{B}(x) := \sqrt{\frac{a}{\mu}} \begin{bmatrix}
    0_{2,3} & B_1(x) & 0_{2,2} & 0_{2,5} \\
    0_{2,3} & 0_{2,4} & 0_{2,2} & B_2(x) \\
    B_3(x) & 0_{2,4} & B_4(x) & B_5(x) \end{bmatrix} \in \mathbb{R}^{6 \times 14},
\end{equation*}
with $0_{i,j}$ the null matrix of $i$ rows and $j$ columns.
We calculate the submatrices using the averaged variational equations for the orbital elements from \citep{trajectory_dynamics}:
\begin{align*}
    B_1(x) &= \begin{bmatrix}  ae & 2a\sqrt{1-e^2} & 0 & 0\\
    \frac{1}{2}(1-e^2) & -\frac{3}{2}e\sqrt{1-e^2} & \sqrt{1-e^2} & -\frac{1}{4}e\sqrt{1-e^2}\end{bmatrix} \\
     B_2(x) &= \begin{bmatrix} \cos{\omega} & 0 \\ 0 & \sin{\omega}\csc{i} \end{bmatrix} \begin{bmatrix} \frac{-3e}{2\sqrt{1-e^2}} & \frac{(1+e^2)}{2\sqrt{1-e^2}} & \frac{-e}{4\sqrt{1-e^2}} & -\frac{1}{2}\tan{\omega} & \frac{1}{4}e\tan{\omega} \\
    \frac{-3e}{2\sqrt{1-e^2}} & \frac{(1+e^2)}{2\sqrt{1-e^2}} & \frac{-e}{4\sqrt{1-e^2}} & \frac{1}{2}\cot{\omega} & -\frac{1}{4}e\cot{\omega} \end{bmatrix} \\
    B_3(x) &= \begin{bmatrix} \sqrt{1-e^2} & -\frac{1}{2e}\sqrt{1-e^2} & 0 \\
    -3 & \frac{3e}{2}+\frac{1}{2e} & -\frac{1}{2}e^2 \end{bmatrix}      B_4(x) = \begin{bmatrix} \frac{1}{2e}(2-e^2) & -\frac{1}{4}  \\
    -\frac{1}{2e}(2-e^2)\sqrt{1-e^2} & \frac{1}{4}\sqrt{1-e^2}\end{bmatrix} \\
    B_5(x) &= \cos{i}\csc{i}\begin{bmatrix}
  \frac{3}{2}e\frac{\sin{\omega}}{\sqrt{1-e^2}} & -\frac{1}{2}(1+e^2)\frac{\sin{\omega}}{\sqrt{1-e^2}} & \frac{1}{4}e\frac{\sin{\omega}}{\sqrt{1-e^2}} & -\frac{1}{2} & \frac{1}{4}e \\
    0 & 0 & 0 & 0 & 0 \end{bmatrix}
\end{align*}
with $\mu = 3.986 \times 10^{14}\, \rm{m^3s^{-2}}$ being the standard gravitational parameter of the Earth.


\bibliographystyle{IEEEtran}
\bibliography{ref}

\begin{thebibliography}{10}
\providecommand{\url}[1]{#1}
\csname url@samestyle\endcsname
\providecommand{\newblock}{\relax}
\providecommand{\bibinfo}[2]{#2}
\providecommand{\BIBentrySTDinterwordspacing}{\spaceskip=0pt\relax}
\providecommand{\BIBentryALTinterwordstretchfactor}{4}
\providecommand{\BIBentryALTinterwordspacing}{\spaceskip=\fontdimen2\font plus
\BIBentryALTinterwordstretchfactor\fontdimen3\font minus
  \fontdimen4\font\relax}
\providecommand{\BIBforeignlanguage}[2]{{%
\expandafter\ifx\csname l@#1\endcsname\relax
\typeout{** WARNING: IEEEtran.bst: No hyphenation pattern has been}%
\typeout{** loaded for the language `#1'. Using the pattern for}%
\typeout{** the default language instead.}%
\else
\language=\csname l@#1\endcsname
\fi
#2}}
\providecommand{\BIBdecl}{\relax}
\BIBdecl

\bibitem{SIAM_CT}
J.-B. Bouvier, K.~Xu, and M.~Ornik, ``Quantitative resilience of linear
  driftless systems,'' in \emph{2021 Proceedings of the Conference on Control
  and its Applications}, 2021, pp. 32 -- 39.

\bibitem{NASA_redundancy}
\BIBentryALTinterwordspacing
R.~C. Suich and R.~L. Patterson, ``How much redundancy: Some cost
  considerations, including examples for spacecraft systems,'' NASA Technical
  Memorandum 103197, Lewis Research Center, Cleveland, Ohio, Tech. Rep., 1990.
  [Online]. Available: \url{https://www.osti.gov/biblio/5894363}
\BIBentrySTDinterwordspacing

\bibitem{partial_LOC}
B.~Xiao, Q.~Hu, and P.~Shi, ``Attitude stabilization of spacecrafts under
  actuator saturation and partial loss of control effectiveness,'' \emph{IEEE
  Transactions on Control Systems Technology}, vol.~21, no.~6, pp. 2251 --
  2263, 2013.

\bibitem{actuator_lock}
G.~Tao, S.~Chen, and S.~M. Joshi, ``An adaptive actuator failure compensation
  controller using output feedback,'' \emph{IEEE Transactions on Automatic
  Control}, vol.~47, no.~3, pp. 506 -- 511, 2002.

\bibitem{fault_tolerance}
S.~S. Tohidi, Y.~Yildiz, and I.~Kolmanovsky, ``Fault tolerant control for
  over-actuated systems: An adaptive correction approach,'' in \emph{2016
  American Control Conference (ACC)}.\hskip 1em plus 0.5em minus 0.4em\relax
  IEEE, 2016, pp. 2530 -- 2535.

\bibitem{hypersonic_fault_tolerant}
Y.~Yu, H.~Wang, and N.~Li, ``Fault-tolerant control for over-actuated
  hypersonic reentry vehicle subject to multiple disturbances and actuator
  faults,'' \emph{Aerospace Science and Technology}, vol.~87, pp. 230 -- 243,
  2019.

\bibitem{IFAC}
J.-B. Bouvier and M.~Ornik, ``Resilient reachability for linear systems,''
  \emph{IFAC-PapersOnLine}, vol.~53, no.~2, pp. 4409 -- 4414, 2020, 21st IFAC
  World Congress.

\bibitem{actuator_attack}
H.~Fawzi, P.~Tabuada, and S.~Diggavi, ``Secure estimation and control for
  cyber-physical systems under adversarial attacks,'' \emph{IEEE Transactions
  on Automatic Control}, vol.~59, no.~6, pp. 1454 -- 1467, 2014.

\bibitem{Tubes1971}
D.~Bertsekas and I.~Rhodes, ``On the minimax reachability of target sets and
  target tubes,'' \emph{Automatica}, vol.~7, pp. 233 -- 247, 1971.

\bibitem{Robust}
S.~Rakovi{\'{c}}, E.~Kerrigan, D.~Mayne, and J.~Lygeros, ``Reachability
  analysis of discrete-time systems with disturbances,'' \emph{IEEE
  Transactions on Automatic Control}, vol.~51, no.~4, pp. 546 -- 561, April
  2006.

\bibitem{journal_paper}
J.-B. Bouvier and M.~Ornik, ``Designing resilient linear systems,'' \emph{IEEE
  Transactions on Automatic Control}, vol.~67, no.~9, pp. 4832 -- 4837, 2022.

\bibitem{water_qr}
S.~Shin, S.~Lee, D.~R. Judi, M.~Parvania, E.~Goharian, T.~McPherson, and S.~J.
  Burian, ``A systematic review of quantitative resilience measures for water
  infrastructure systems,'' \emph{Water}, vol.~10, no.~2, pp. 164 -- 189, 2018.

\bibitem{nuclear_pp_qr}
J.~T. Kim, J.~Park, J.~Kim, and P.~H. Seong, ``Development of a quantitative
  resilience model for nuclear power plants,'' \emph{Annals of Nuclear Energy},
  vol. 122, pp. 175 -- 184, 2018.

\bibitem{max-min_programming}
M.~E. Posner and C.-T. Wu, ``Linear max-min programming,'' \emph{Mathematical
  Programming}, vol.~20, no.~1, pp. 166 -- 172, 1981.

\bibitem{semi-infinite_programming}
R.~Hettich and K.~O. Kortanek, ``Semi-infinite programming: theory, methods,
  and applications,'' \emph{SIAM Review}, vol.~35, no.~3, pp. 380 -- 429, 1993.

\bibitem{liberzon}
D.~Liberzon, \emph{Calculus of Variations and Optimal Control Theory: a Concise
  Introduction}.\hskip 1em plus 0.5em minus 0.4em\relax Princeton University
  Press, 2011.

\bibitem{Neustadt}
L.~W. Neustadt, ``The existence of optimal controls in the absence of convexity
  conditions,'' \emph{Journal of Mathematical Analysis and Applications},
  vol.~7, pp. 110 -- 117, 1963.

\bibitem{Maxmax_Minimax_Quotient_thm}
J.-B. Bouvier and M.~Ornik, ``The maximax minimax quotient theorem,''
  \emph{Journal of Optimization Theory and Applications}, vol. 192, pp. 1084 --
  1101, 2022.

\bibitem{feedback_linearization}
A.~Isidori, \emph{Nonlinear Control Systems, An Introduction}.\hskip 1em plus
  0.5em minus 0.4em\relax Springer Verlag, 1989.

\bibitem{Borgest}
W.~Borgest and P.~Varaiya, ``Target function approach to linear pursuit
  problems,'' \emph{IEEE Transactions on Automatic Control}, vol.~16, no.~5,
  pp. 449 -- 459, 1971.

\bibitem{Sakawa}
Y.~Sakawa, ``Solution of linear pursuit-evasion games,'' \emph{SIAM Journal on
  Control}, vol.~8, no.~1, pp. 100 -- 112, 1970.

\bibitem{LaSalle}
J.~LaSalle, ``Time optimal control systems,'' \emph{Proceedings of the National
  Academy of Sciences of the United States of America}, vol.~45, no.~4, pp. 573
  -- 577, 1959.

\bibitem{Sussmann}
H.~J. Sussmann, ``A bang-bang theorem with bounds on the number of
  switchings,'' \emph{SIAM Journal on Control and Optimization}, vol.~17,
  no.~5, pp. 629 -- 651, 1979.

\bibitem{Aronsson}
G.~Aronsson, ``Global controllability and bang-bang steering of certain
  nonlinear systems,'' \emph{SIAM Journal on Control}, vol.~11, no.~4, pp. 607
  -- 619, 1973.

\bibitem{Glashoff}
K.~Glashoff and E.~Sachs, ``On theoretical and numerical aspects of the
  bang-bang-principle,'' \emph{Numerische Mathematik}, vol.~29, no.~1, pp. 93
  -- 113, 1977.

\bibitem{inf_dim_analysis}
C.~Aliprantis and K.~Border, \emph{Infinite Dimensional Analysis: A
  Hitchhiker's Guide}.\hskip 1em plus 0.5em minus 0.4em\relax New York:
  Springer, 2006.

\bibitem{ko15}
\BIBentryALTinterwordspacing
D.~Kolosa, ``Implementing a linear quadratic spacecraft attitude control
  system,'' Master's thesis, Western Michigan University, 2015. [Online].
  Available: \url{https://scholarworks.wmich.edu/masters_theses/661}
\BIBentrySTDinterwordspacing

\bibitem{trajectory_dynamics}
J.~S. Hudson and D.~J. Scheeres, ``Reduction of low-thrust continuous controls
  for trajectory dynamics,'' \emph{Journal of Guidance, Control, and Dynamics},
  vol.~32, no.~3, pp. 780 -- 787, 2009.

\bibitem{unstable_quadcopters}
A.~Freddi, A.~Lanzon, and S.~Longhi, ``A feedback linearization approach to
  fault tolerance in quadrotor vehicles,'' \emph{IFAC proceedings volumes},
  vol.~44, no.~1, pp. 5413 -- 5418, 2011.

\bibitem{UAV_equations}
V.~G. Adir and A.~M. Stoica, ``Integral {LQR} control of a star-shaped
  octorotor,'' \emph{Incas Bulletin}, vol.~4, no.~2, pp. 3 -- 18, 2012.

\bibitem{drone_model}
H.~Romero, S.~Salazar, A.~Sanchez, and R.~Lozano, ``A new {UAV} configuration
  having eight rotors: Dynamical model and real-time control,'' in \emph{46th
  IEEE Conference on Decision and Control}.\hskip 1em plus 0.5em minus
  0.4em\relax IEEE, 2007, pp. 6418 -- 6423.

\bibitem{ADMIRE_1}
S.~T.~G. Ola~H{\"{a}}rkeg{\aa}rd, ``Resolving actuator redundancy - optimal
  control vs. control allocation,'' \emph{Automatica}, vol.~41, pp. 137 -- 144,
  2005.

\bibitem{UAV_prop}
L.~Wu, Y.~Ke, and B.~Chen, ``Systematic modeling of rotor dynamics for small
  unmanned aerial vehicles,'' in \emph{International Micro Air Vehicle
  Competition and Conference}, 2016, pp. 284 -- 290.

\bibitem{Eaton}
J.~Eaton, ``An iterative solution to time-optimal control,'' \emph{Journal of
  Mathematical Analysis and Applications}, vol.~5, no.~2, pp. 329 -- 344, 1962.

\end{thebibliography}

\end{document}